\documentclass[anonymous,runningheads,lineno]{llncs}

\usepackage[utf8]{inputenc} 
\usepackage[T1]{fontenc}    
\usepackage{hyperref}       
\usepackage{url}            
\usepackage{booktabs}       
\usepackage{amsfonts}       
\usepackage{nicefrac}       
\usepackage{microtype}      
\usepackage{xcolor}         
\usepackage{soul}

\usepackage{microtype}
\usepackage{graphicx}
\usepackage{subfigure}
\usepackage{booktabs} 

\usepackage{hyperref}

\usepackage{amsmath}
\usepackage{amssymb}
\usepackage{mathtools}

\usepackage{multirow}
\usepackage{epsfig}
\usepackage{graphicx}
\usepackage{color}
\usepackage{bigints}
\usepackage{array}
\usepackage[normalem]{ulem}
\usepackage{xspace}
\usepackage{algorithm}
\usepackage{algorithmic}
\usepackage{enumerate}
\usepackage{hyperref}
\usepackage{numprint}
\usepackage{changepage}
\usepackage{enumerate}
\usepackage{wrapfig}
\usepackage{subfigure}
\usepackage{caption}

\usepackage{lineno}

\newcommand{\real}{\mathbb{R}}

\newcommand{\eps}{\varepsilon}
\newcommand{\mat}[1] {\mathtt{#1}}
\newcommand{\supp} {\mathrm{supp}}

\newcommand{\iprod}[1]  {\bigl\langle #1 \,\bigr\rangle}
\newcommand{\norm}[1]   {\| #1 \|}
\newcommand{\cancel}[1] {}

\newbox\ProofSym
\setbox\ProofSym=\hbox{%
\unitlength=0.18ex%
\begin{picture}(10,10)
\put(0,0){\framebox(9,9){}}
\put(0,3){\framebox(6,6){}}
\end{picture}}

\usepackage[capitalize,noabbrev]{cleveref}

\usepackage[textsize=tiny]{todonotes}

\begin{document}
\title{A Dynamic Working Set Method for \\ Compressed Sensing\thanks{Research supported by Research Grants Council, Hong Kong, China (project no.~16203718).  The conference version is to appear in Proceedings of the International Computing and Combinatorics Conference, 2025.}}

\titlerunning{A Dynamic Working Set Method}
\author{Siu-Wing Cheng \and Man Ting Wong}

\authorrunning{S.-W. Cheng and M.T. Wong}
\institute{HKUST, Hong Kong, China}

\maketitle              

\begin{abstract}
We propose a dynamic working set method (DWS) for the problem $\min_{\mat{x} \in \real^n} \frac{1}{2}\norm{\mat{Ax}-\mat{b}}^2 + \eta\norm{\mat{x}}_1$ that arises from compressed sensing.  DWS manages the working set while iteratively calling a regression solver to generate progressively better solutions.  Our experiments show that DWS is more efficient than other state-of-the-art software in the context of compressed sensing.  Scale space such that $\norm{\mat{b}}=1$.  Let $s$ be the number of non-zeros in the unknown signal.  We prove that for any given $\eps > 0$, DWS reaches a solution with an additive error $\eps/\eta^2$ such that each call of the solver uses only $O(\frac{1}{\eps}s\log s \log\frac{1}{\eps})$ variables, and each intermediate solution has $O(\frac{1}{\eps}s\log s\log\frac{1}{\eps})$ non-zero coordinates.

\end{abstract}

\section{Introduction}

Compressed sensing allows for the recovery of sparse signals using very few observations.  Applications include multislice brain imaging~\cite{mri}, wavelet-based image/signal reconstruction and restoration~\cite{nosr}, the single-pixel Camera~\cite{spiv}, and hyperspectral imaging~\cite{ssec}.  There are two components in compressed sensing.  First, a matrix $\mat{A} \in \real^{k \times n}$ is designed such that for any unknown signal $\mat{z} \in \real^n$, a small number of $k$ noisy observations are taken as $\mat{b} = \mat{Az} + \mat{n} \in \real^k$, where $\mat{n}$ denotes Gaussian noise.  Second, an algorithm is run on $\mat{A}$ and $\mat{b}$ to recover $\mat{z}$.

Let $s$ be the number of non-zeros in the unknown $\mat{z} \in \real^n$.
In many applications, $s$ is no more than 8\% of $n$ (e.g.~\cite{spiv,ssec}), and it has been argued~\cite{cs} that certain images with $n$ pixels can be reconstructed with $O(\sqrt{n}\log^3 n)$ observations, i.e., $s = o(n)$.  If $\mat{A}$ has the \emph{restricted isometry properties} (RIP), it has been proved that $\mat{z}$ can be recovered with high probability by solving 
\begin{equation}
\min_{\mat{x} \in \real^n} F(\mat{x}) = \min_{\mat{x} \in \real^n} \frac{1}{2}\|\mat{A}\mat{x} - \mat{b}\|^2 + \eta \|\mat{x}\|_1
\label{eq:cs}
\end{equation}
for an appropriate $\eta > 0$ with $k = Cs\ln(n/s)$ for some constant $C$~\cite{bar08,rupe,cs}.  It is popular to use a random matrix $\mat{A}$ to achieve RIP with high probability. 
For example, sample each matrix entry independently from the normal distribution $\mathcal{N}(0,1)$ and then orthonormalize the rows~\cite{gpsr};
all non-zero singular values of $\mat{A}$ are thus equal to $1$.  A detailed discussion of RIP can be found in~\cite{bar08,rupe,cs}.

In this paper, we are concerned with solving $\min_{x \in \real^n} F(\mat{x})$ when $s \ll n$, $\mat{A}$ is an arbitrary $k \times n$ matrix with $\norm{\mat{A}} \leq 1$, and $\eta = \alpha\norm{\mat{A}^t\mat{b}}_\infty$ for some fixed $\alpha \in (0,1)$.\footnote{Whenever $\eta \geq \norm{\mat{A}^t\mat{b}}_\infty$, $\mat{x} = 0$ is the optimal solution~\cite{fuchs}.}  We propose a dynamic working set method and show that it gives superior performance than several state-of-the-art solvers in compressed sensing experiments when $\mat{A}$ is generated randomly as described above.  We also mathematically analyze the convergence and efficiency of our method.

\vspace{6pt}

\noindent {\bf Related work.}  If $\mat{A}$ in \eqref{eq:cs} is an arbitrary matrix, the problem is generally known as Lasso~\cite{lasso}, which is originally proposed for regularized regression and variable selection.   The sparsity level for Lasso to yield the best fit is typically unknown, whereas the compressed sensing applications often give a specific sparsity range for the unknown signal.  Problem~\eqref{eq:cs} can be transformed to a convex quadratic programming problem (e.g.~\cite{gpsr}) that can be solved in $O(n^3L)$ time~\cite{monteiro89}, where $L$ is the total number of bits representing the instance.  Tailor-made algorithms have also been developed.  The earlier ones include gradient projection for sparse reconstruction (GPSR)~\cite{gpsr}, iterated thresholding (IST)~\cite{dd04}, L1\_LS~\cite{kse07}, the homotopy method~\cite{ddy08}, and L1-magic~\cite{l1magic}.
In compressed sensing experiments, L1\_LS runs faster than L1-magic and the homotopy method~\cite{l1magic}, and GPSR runs faster than IST and L1\_LS~\cite{gpsr}.

Recently, coordinate descent algorithms with theoretical guarantees have been effective in solving large convex optimization problems with sparse solution~\cite{nesterov2012efficiency,tseng2009coordinate}.  Two solvers in this category are glmnet~\cite{glmlassopaper} and scikit-learn~\cite{scikit}.  To solve problems with even more variables, 
working set strategies have been combined with coordinate descent or other solvers.  
They iteratively call a solver to generate progressively better solutions, and a small set of free variables is maintained to reduce the execution time of each call.
Algorithms that employ the working set methods include Picasso~\cite{picasso}, Blitz~\cite{Blitz}, Fireworks~\cite{fireworks}, Celer~\cite{celer2}, and Skglm~\cite{skglm}.  The convergence of these methods has been proven.  
In Lasso experiments, Blitz runs faster than L1\_LS and glmnet~\cite{Blitz}, Celer runs faster than Blitz and scikit-learn~\cite{celer2}, and Skglm performs better than Celer, Blitz, and Picasso~\cite{skglm}.  

According to the literature, GPSR, Skglm, and Celer would be the major competing solvers for compressed sensing problems.

\vspace{6pt}

\noindent {\bf Our contributions.}  We propose a dynamic working set (DWS) algorithm for solving problem~\eqref{eq:cs} when $s \ll n$, $\mat{A}$ is an arbitrary $k \times n$ matrix with $\norm{\mat{A}} \leq 1$, and $\eta =\alpha\norm{\mat{A}^t\mat{b}}_\infty$ for a fixed $\alpha \in (0,1)$.  

Define the \emph{support set} of a solution to be the subset of non-zero variables in it.  DWS checks how well the support set size matches the working set size in the previous iteration.  The result determines the number of free variables that will be added to the previous support set to form the next working set.

We ran compressed sensing experiments on DWS with GPSR as the solver.  We set $s$ to be 1\%, 4\%, and 8\% of $n$ which is similar to the ranges of $s$ used in previous works~\cite{fsso,gpsr,srfr}. DWS is $1.91\times$ faster than Skglm, $3\times$ faster than Celer, and $2.45\times$ faster than running GPSR alone on average.  Similar trends are observed for other values of $s$ in the range of 1\% to 8\% of $n$.

Scale space such that $\norm{\mat{b}} = 1$.  Take any $\eps \in (0,1)$.  Let $U$ be an upper bound on any working set size before DWS reaches a solution $\mat{x}_r$ such that $F(\mat{x}_r) \leq \mathrm{optimum} + \eps/\eta^2$.  We prove that $U = O(\frac{1}{\eps}s\log s\log\frac{\eta}{\eps})$ if $\eps$ is given beforehand and $U = O(\frac{1}{\eps}k\log k\log\frac{\eta}{\eps})$ otherwise.  There are two implications.  First, DWS can converge to any positive error.  Second, if $\eps$ is given beforehand or $k = \Theta(s\log(n/s))$ (which allows the recovery of the sparse signal), then DWS uses provably small working sets and produces provably sparse solutions until $\mat{x}_r$.

\vspace{6pt}

\noindent {\bf Notations.}  Matrices are represented by uppercase letters in typewriter font.  Vectors are represented by lowercase letters in typewriter font or lowercase Greek symbols. The inner product of $\mat{x}$ and $\mat{y}$ is $\langle \mat{x},\mat{y} \rangle$ or $\mat{x}^t\mat{y}$. We use $(\mat{x})_i$ to denote the $i$-th coordinate of a vector $\mat{x}$. 
Define the support set of $\mat{x}$ to be $\mbox{supp}(\mat{x}) = \bigl\{i : (\mat{x})_i \neq 0 \bigr\}$.  
Given a matrix $\mat{M}$ and a vector $\mat{x}$, we use $\norm{\mat{M}}$ and $\norm{\mat{x}}$ to denote their $L_2$-norms, and we use $\norm{\mat{x}}_1$ and $\norm{\mat{x}}_\infty$ to denote the $L_1$-norm and $L_\infty$-norm of $\mat{x}$, respectively. Let $n$ be the total number input variables.  Let $s$ be the support set size of the optimal solution.

\section{Algorithm DWS}
\label{sec:alg}

Let $f(\mat{x}) = \frac{1}{2}\|\mat{Ax}-\mat{b}\|^2$.  Let $g(\mat{x}) = \eta\norm{\mat{x}}_1$.  The objective function is $F(\mat{x}) = f(\mat{x}) + g(\mat{x})$.  DWS calls a solver iteratively. In each iteration, some variables are free, forming the \emph{working set}, and the others are fixed at zero.  We use $\mat{x}_r$ to denote the solution returned by the solver in the $r$-th iteration.

Algorithm~\ref{alg:1} gives the pseudocode of DWS.  We define $\mat{x}_{0} = \mat{0}$.  For $r \geq 0$, we extract a subset of variables 
\[
E_{r} = \left\{j \in [n] : \left|\frac{\partial f(\mat{x}_{r})}{\partial (\mat{x})_j}\right| > \eta \right\}.  
\]
We will prove that for all $j \in E_{r}$, if $\partial f(\mat{x}_{r})/\partial (\mat{x})_j < 0$, the $j$-th positive axis is a descent direction from $\mat{x}_{r}$; otherwise, if $\partial f(\mat{x}_{r})/\partial (\mat{x})_j > 0$, the $j$-th negative axis is a descent direction from $\mat{x}_{r}$.  The \emph{weight} of $j \in [n]$ is $\left|\partial f(\mat{x}_{r}/\partial (\mat{x})_j\right|$.  An element $j$ is \emph{heavier} than another if its weight is larger.  DWS uses a parameter $p_0$ to initialize the first working set $W_1$ to consist of  the $p_0$ elements of $[n]$ with the $p_0$ largest $\bigl|\partial f(\mat{x}_0)/\partial (\mat{x})_j\bigr|$.
When $p_0 \leq |E_0|$, $W_1$ consists of the $p_0$ heaviest elements of $E_0$, and the same initialization is done in Skglm. When $p_0 > |E_0|$, Skglm selects $p_0-|E_0|$ variables outside $E_0$ in some order and inserts them into $W_1$, which is similar to what we do. Celer also starts with a working set of size $p_0$ by some selection criterion.  The working set of DWS for the $(r+1)$-th iteration for $r \geq 1$ is $W_{r+1} = \supp(\mat{x}_{r}) \cup \{\text{the $\tau_{r+1}$ heaviest elements in $E_{r}$}\}$,
where $\tau_{r+1}$ is defined in lines~\ref{alg:tau-1}--\ref{alg:tau-3} of Algorithm~\ref{alg:1}.   DWS uses a basic step size $\tau$ for increasing the working set size, and $\tau_{r+1}$ is equal to $\min\{h^{a_r}\tau, k, |E_r|\}$ for some appropriate integer $a_r$.  By our assumption that $k > s$, we will not release more than $k$ variables from $E_{r}$ to $W_{r+1}$.  The variables in $W_r$ that are zero will be kicked out of $W_{r+1}$.  This can significantly reduce the running time of the next iteration. The rationale behind the setting of $a_r$ is:
\begin{itemize}

\item If $|\supp(\mat{x}_r)| \leq |\supp(\mat{x}_{r-1})| + \tau/h$, lines~\ref{alg:tau-1}--\ref{alg:tau-3} of Algorithm~\ref{alg:1} set $a_r = 0$, i.e., $\tau_{r+1} = \min\{\tau,k,|E_r|\}$.  The slow growth in the support set size suggests that the working set size may be close to the ideal.  We should not increase the working set size so much to slow down the next iteration.
 
\item Otherwise, let $m$ be the smallest non-negative integer such that $|\supp(\mat{x}_r)| \leq |\supp(\mat{x}_{r-1})| + h^m \tau$.  We can release $h^{m+1}\tau$ or $h^{a_{r-1}+1}\tau$ variables from $E_r$, i.e., a factor $h$ more.
To avoid a large increase in the working set size, lines~\ref{alg:tau-1}--\ref{alg:tau-3} set $a_r = \min\{m+1,a_{r-1}+1\}$.
\end{itemize}

\begin{algorithm}
\caption{DWS}\label{alg:1}
\begin{algorithmic}[1]
        \STATE{$h \gets$ any constant in $(1,2]$} \hspace*{31pt} /* $h = 2$ in the experiments. */
        \STATE{$\tau \gets$ any integer in $[k]$} \hspace*{47pt} /* $\tau = \lfloor 4\ln^2 n \rfloor$ in the experiments */
        \STATE{$\mat{x}_0 \gets 0$}
        \STATE{compute $\nabla f(\mat{x}_{0}) = -\mat{A}^t\mat{b}$ to generate $E_0$}  
        \STATE{$W_1 \gets \left\{\text{the $p_0$ elements of $[n]$ with the $p_0$ largest $\bigl|\frac{\partial f(\mat{x}_0)}{\partial (\mat{x})_j}\bigr|$}\right\}$}  \hspace{10pt} /* $p_0=10$ in the experiments */
        \STATE{$a_0 \gets 0$; $r \gets 1$}
	\WHILE{$E_{r-1} \not= \emptyset$}
            \STATE{$\mat{A}_{r} \gets$ submatrix of $\mat{A}$ with columns corresponding to $W_{r}$}
	    \STATE{$\mat{x}_r \gets \text{optimal solution obtained by calling the solver with $\mat{A}_{r}$ and $\mat{b}$}$} \label{alg:qp}
            \STATE{$m \gets \text{the smallest integer in $[-1,\infty)$ s.t. $|\supp(\mat x_r)|\le h^m \tau+ |\supp(\mat x_{r-1})|$}$} \label{alg:tau-1}
            \STATE{$a_r \gets \min\{m+1,a_{r-1}+1\}$}
            \STATE{compute $\nabla f(\mat{x}_{r}) = \mat{A}^t\mat{A}\mat{x}_r -\mat{A}^t\mat{b}$ to generate $E_{r}$}
            \STATE{$\tau_{r+1} \gets \min\{h^{a_r}\tau, \, k, \, |E_r|\}$}   \label{alg:tau-3}
            \STATE{$W_{r+1} \gets \supp(\mat{x}_{r}) \cup \{\text{the $\tau_{r+1}$ heaviest elements in $E_{r}$}\}$} \label{alg:update}
            \STATE{$r \gets r+1$}
        \ENDWHILE
        \STATE{{\bf return} $\mat{x}_r$}
\end{algorithmic}
\end{algorithm}

\section{Experimental results}

In our experiments, we generate a random matrix as described in the introduction.  All non-zero singular values of $\mat{A}$ are equal to 1.  To generate a vector $\mat{b}$, we first generate a true signal $\mat{z} \in \mathbb{R}^n$ by sampling $s$ coordinates uniformly at random, setting each to $-1$ or $1$ with probability 1/2, and setting the other $n-s$ coordinates to zero.  Then, compute $\mat{b} = \mat{A}\mat{z} + \mat{n}$, where each entry of $\mat{n}$ is drawn independently from $\mathcal{N}(0,10^{-4})$.

We follow the experimental set up in GPSR~\cite{gpsr} to set $\eta = 0.1 \cdot \|\mat{A}^t \mat{b}\|_{\infty}$.  Note that if $\eta \geq \norm{\mat{A}^t\mat{b}}_{\infty}$, then $\mat{x}=0$ is the optimal solution~\cite{fuchs}.  We will report our experimental results with $n \in \{15000,\,30000,\,45000,\,60000\}$, $s \in \{0.01n, 0.04n, 0.08n\}$, and $k = 2s\ln(n/s)$.  
A similar range of $s$ has been used in previous 
works~\cite{fsso,gpsr,srfr} 
and some compressed sensing applications such as Single Pixel Camera \cite{spiv} and hyperspectral imaging \cite{ssec}.  
%
We also tried random inputs with $k = Cs\ln(n/s)$ for $C \in \{1.6,3,4\}$ and other values of $s$ in the range $[0.01n,0.08n]$.
Similar trends have been observed.  
All experiments were run on a 12th Gen Intel\(^{\tiny\textregistered}\) Core\(^{\tiny\texttrademark}\) i9-12900KF CPU (3.19 GHz and 64 GB RAM).



We use BenchOpt~\cite{benchopt} to conduct experiments. It comes with Celer and Skglm.  It allows the user to add new methods.  It generates informative graphs, such as the support set size against iteration, the working set size against iteration, and the \emph{suboptimality curve}, i.e., $F(\mat{x}_r) - F(\mat{x}_*)$ against running time.  

BenchOpt does not simply run a working set method $\mathcal{A}$ to completion.  It starts with a variable $i=1$, runs the first $i$ iterations of $\mathcal{A}$, produces a data point, increments $i$, and repeats the above on the same input. For example, 
a data point for the suboptimality curve is the tuple formed by the running time of the $i$ iterations and $F(\mat{x}_i) - F(\mat{x}_*)$.  
%
As mentioned in~\cite{skglm}, different runs of a solver on the same input may have different running times.  So a plot for $\mathcal{A}$ may not be monotone with respect to the $x$-axis (e.g.~the suboptimality curves for Skglm and Celer in Figure~\ref{fig:sota_showcase_12}).
BenchOpt does not use the termination condition prescribed by $\mathcal{A}$; instead, it stops running $\mathcal{A}$ when the objective function value does not decrease for several consecutive iterations.  The final error is thus clear for comparison.  For clarity, we circle the data points in all graphs at which the corresponding methods should have terminated.   BenchOpt uses the smallest objective function value $V$ among all solvers tested and take $F(\mat{x}_*)$ to be $V - 10^{-10}$.

In implementing DWS, we use the GPSR-BB version of the GPSR package as the solver.  For simplicity, we refer to the GPSR-BB version as GPSR.
Figure~\ref{fig:rtaugpsr_34} shows that DWS is significantly faster than GPSR when $s$ is 1\% or 4\% of $n$; DWS has a similar efficiency as GPSR when $s$ is 8\% of $n$; the average speedup achieved by DWS is roughly $2.45\times$.  

\begin{figure}[]
	\centering
	\includegraphics[width=0.35\columnwidth]{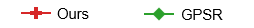}\\
	\includegraphics[width=.32\columnwidth]{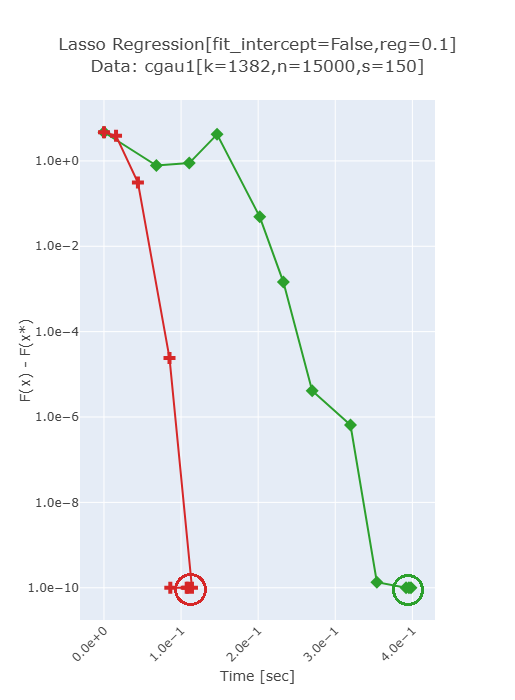}
	\includegraphics[width=.32\columnwidth]{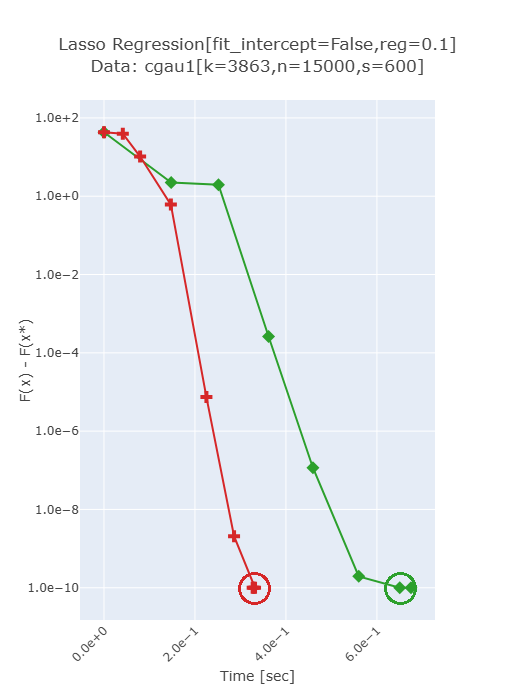}
	\includegraphics[width=.32\columnwidth]{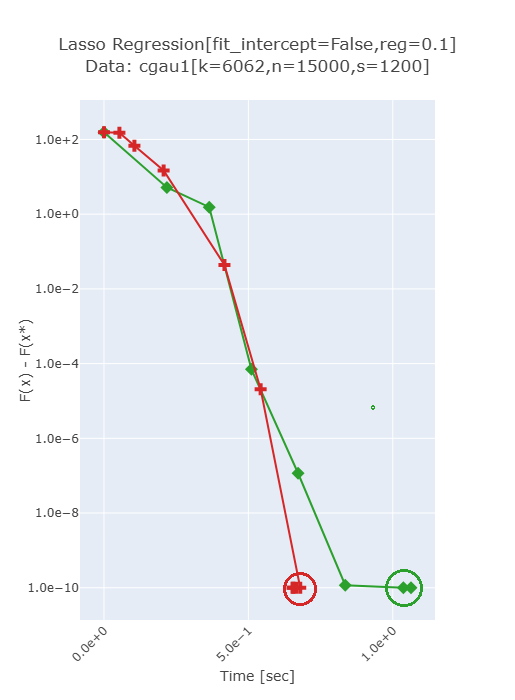}\\
	\includegraphics[width=.32\columnwidth]{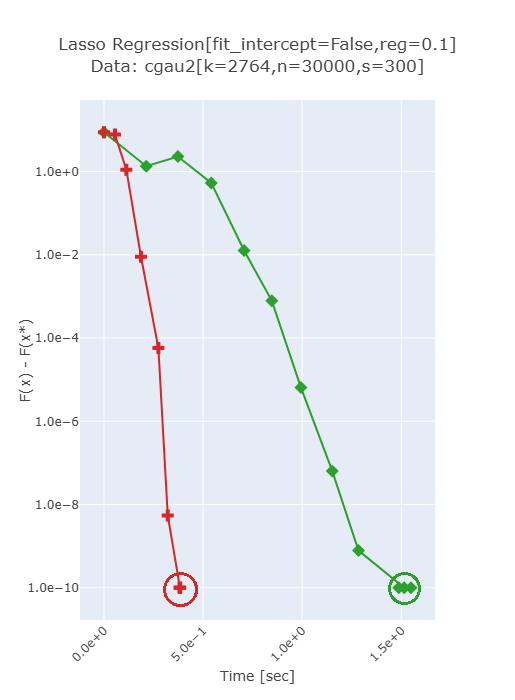}
	\includegraphics[width=.32\columnwidth]{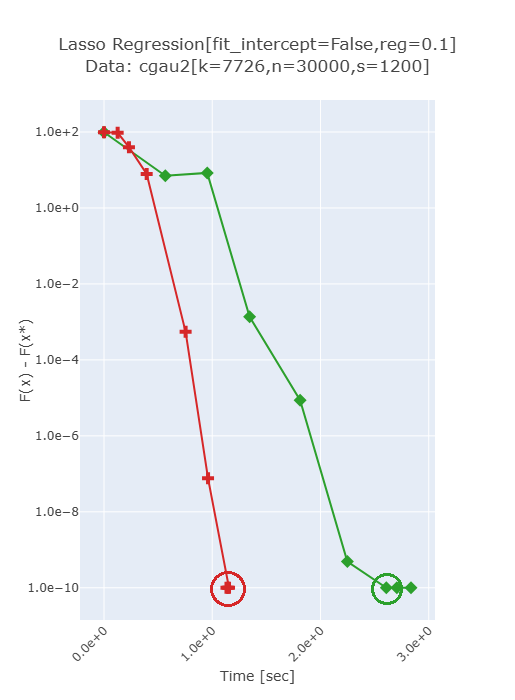}
	\includegraphics[width=.32\columnwidth]{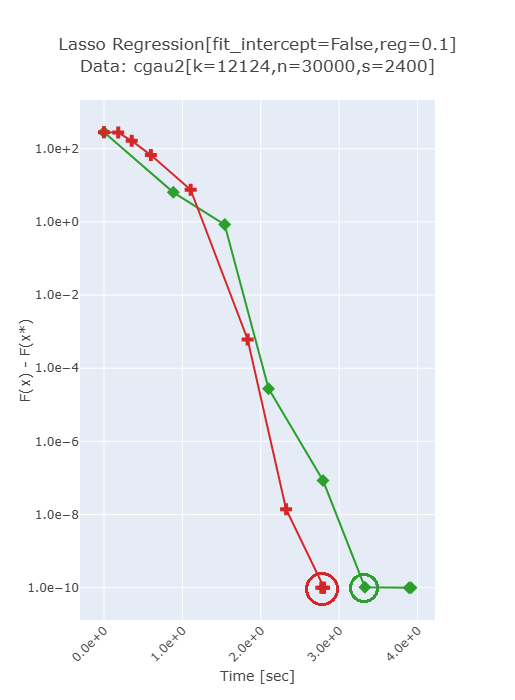}\\
	\includegraphics[width=.32\columnwidth]{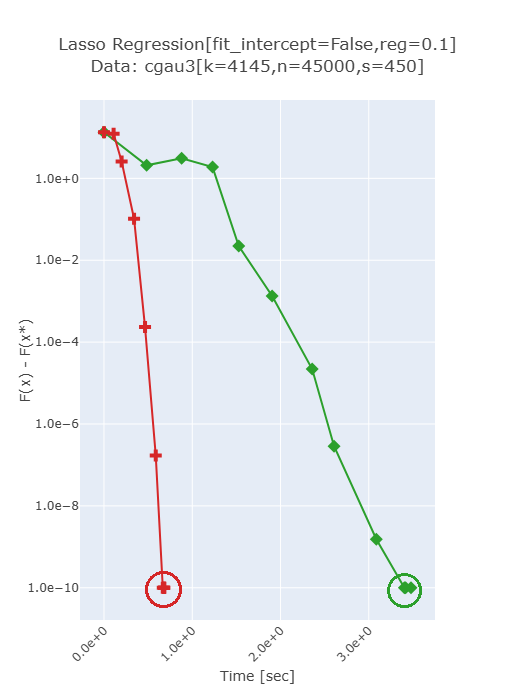}
	\includegraphics[width=.32\columnwidth]{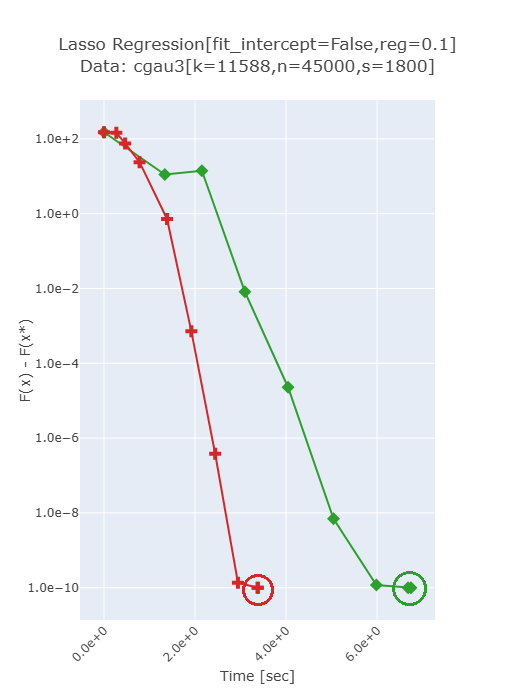}
	\includegraphics[width=.32\columnwidth]{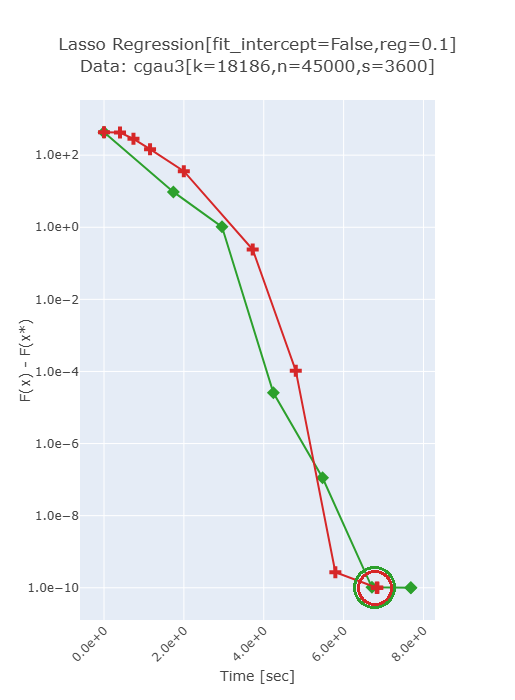}\\
	\includegraphics[width=.32\columnwidth]{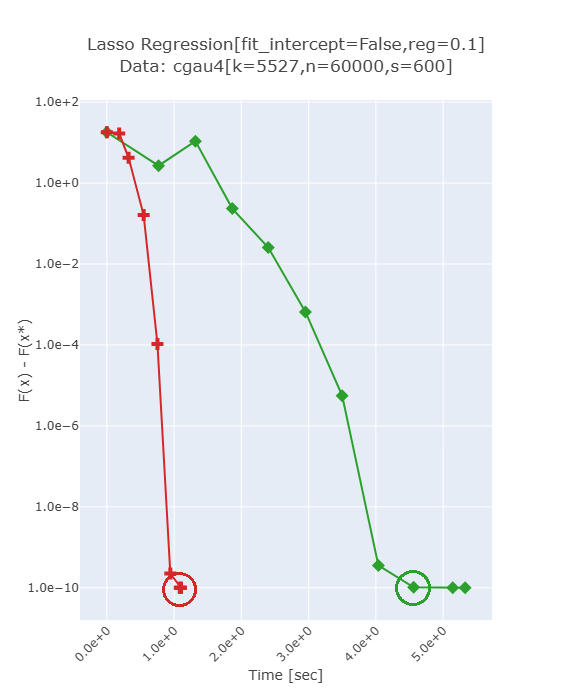}
	\includegraphics[width=.32\columnwidth]{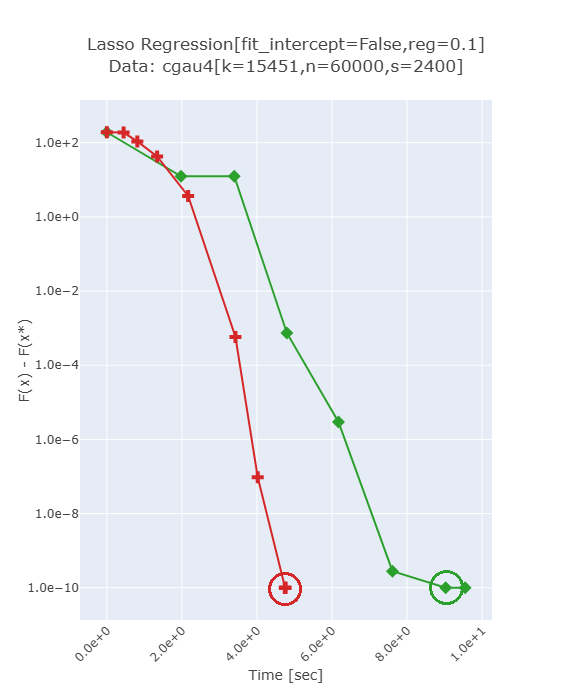}
	\includegraphics[width=.32\columnwidth]{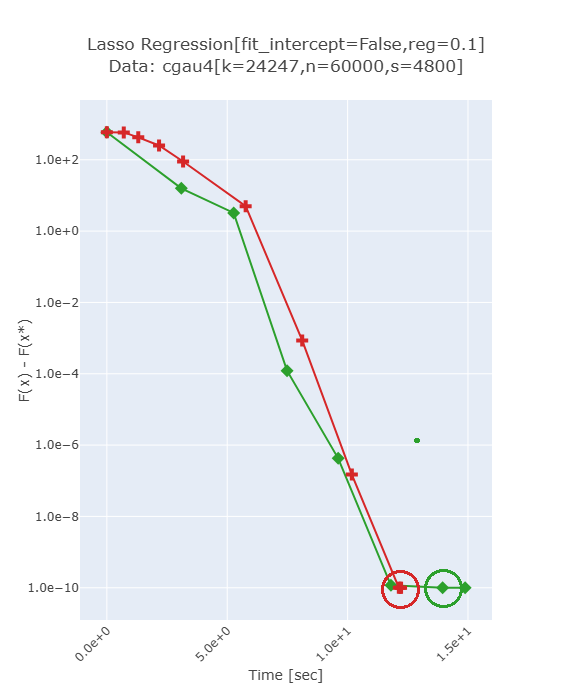}\\
	\caption{Plots of $F(\mat{x}_r) - F(\mat{x}_*)$ against running time for our method and GPSR.}
	\label{fig:rtaugpsr_34}
\end{figure}

Skglm and Celer update the working set using a doubling strategy~\cite{skglm,celer2} that sets the working set size for iteration $r+1$ to be $2 \cdot |\supp(\mat x_r)|$. The variables in the working set for iteration $r$ that are zero may be excluded from the working set for iteration $r+1$.  Celer also supports a non-pruning mode that sets the working set size for iteration $r+1$ to be twice the working set size for iteration $r$, and all variables in the working set for iteration $r$ are kept. In our setting, as shown in Figure~\ref{fig:celerNP},
Celer is not more efficient in the non-pruning mode as $s$ increases. Therefore, we will ignore the non-pruning mode of Celer.\footnote{For Skglm, there is a discrepancy between the doubling strategies in the publicly available code and the paper.  Our description follows the code.  The convergence of Skglm is proved for the version in the paper that grows a working set monotonically.  The convergence of Celer is proved for its non-pruning mode.} 

\begin{figure}[]
	\centering
	\includegraphics[width=0.40\columnwidth]{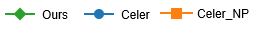}\\
	\includegraphics[width=.32\columnwidth]{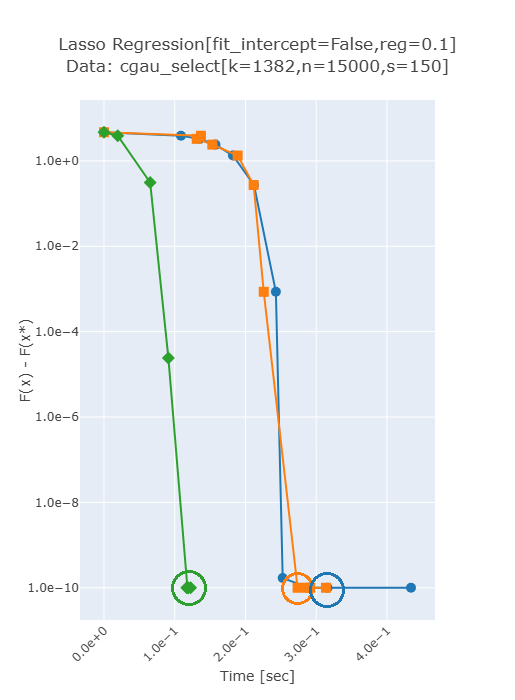}
	\includegraphics[width=.32\columnwidth]{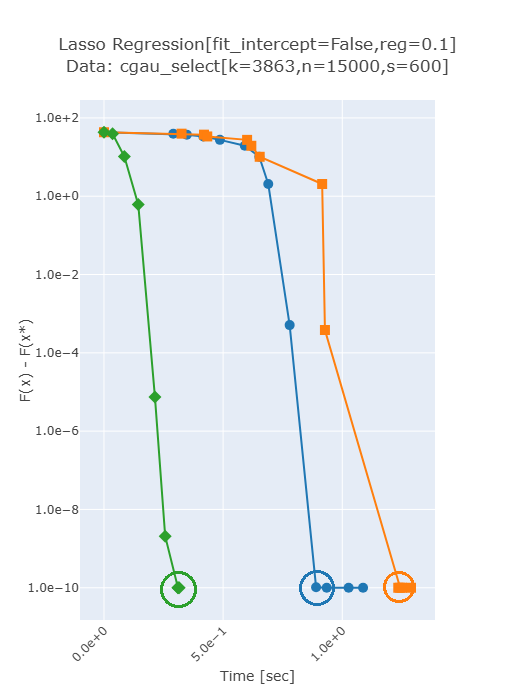}
	\includegraphics[width=.32\columnwidth]{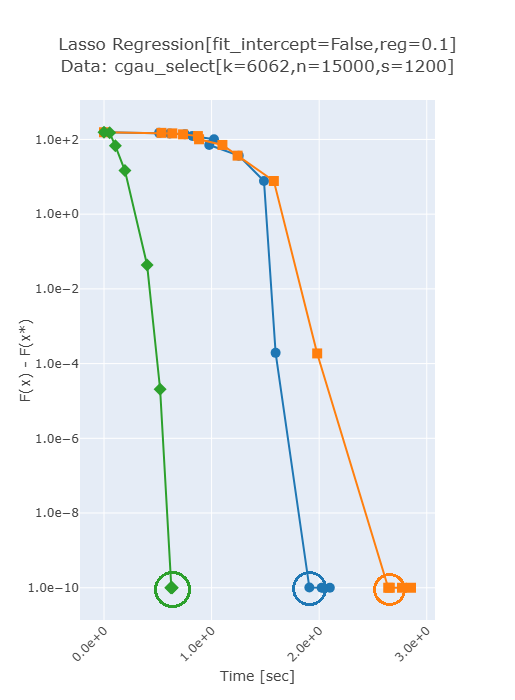}\\
	\includegraphics[width=.32\columnwidth]{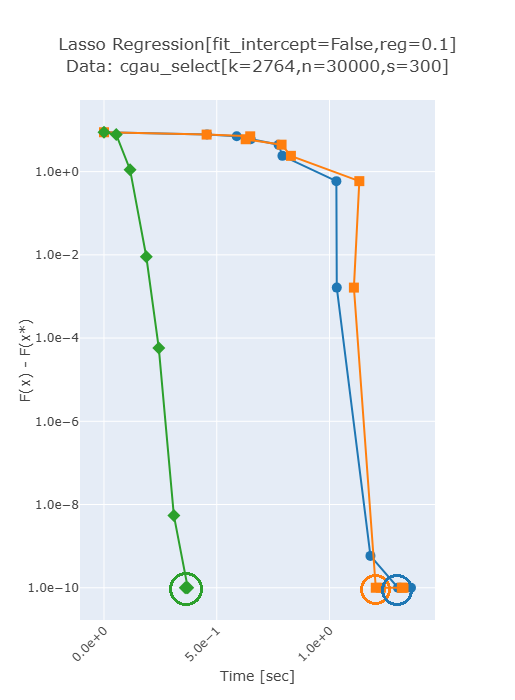}
	\includegraphics[width=.32\columnwidth]{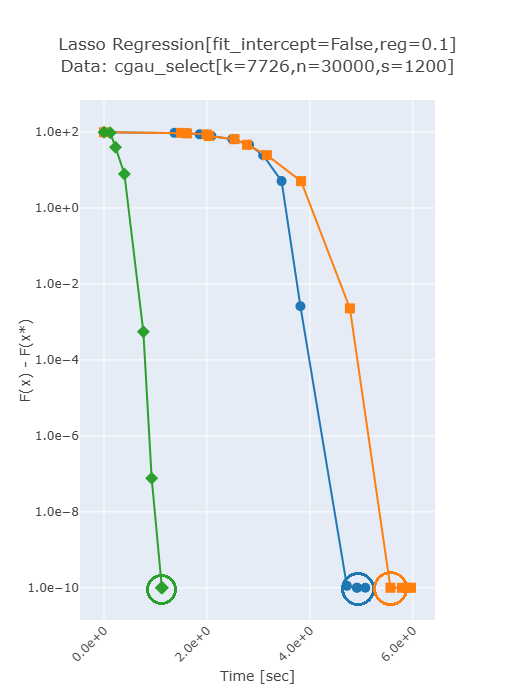}
	\includegraphics[width=.32\columnwidth]{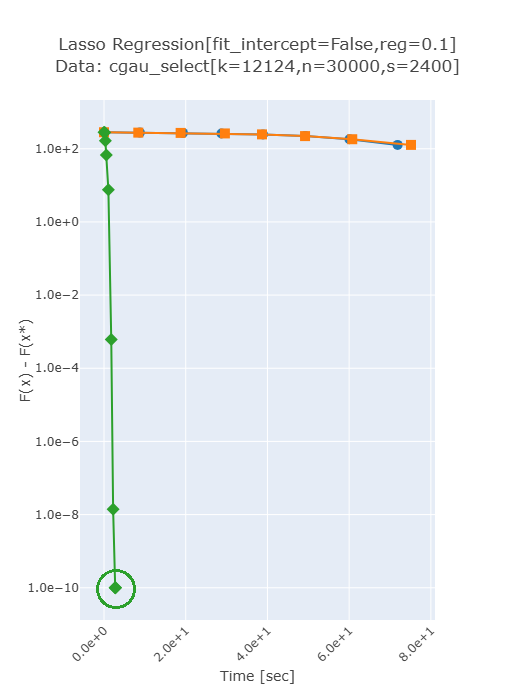}\\
	\caption{Plots of $F(\mat{x}_r) - F(\mat{x}_*)$ against running time, comparing the pruning mode versus the non-pruning mode of Celer for $n \in \{15000, 30000\}$. Our method is also included as a baseline for clearer visualization because Celer does not always converge within the time limit as $n$ and $s$ increase.}
	\label{fig:celerNP}
\end{figure}

\begin{figure}[t]
    \centering
    \includegraphics[width=0.35\columnwidth]{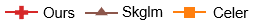}\\
    \includegraphics[width=.32\columnwidth]{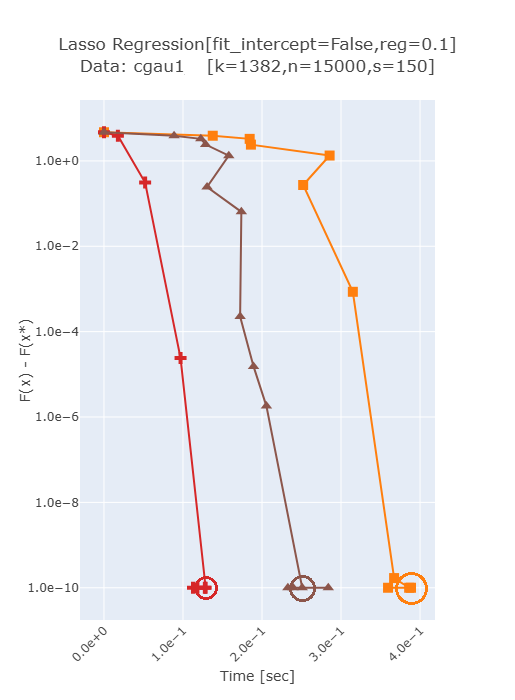}
    \includegraphics[width=.32\columnwidth]{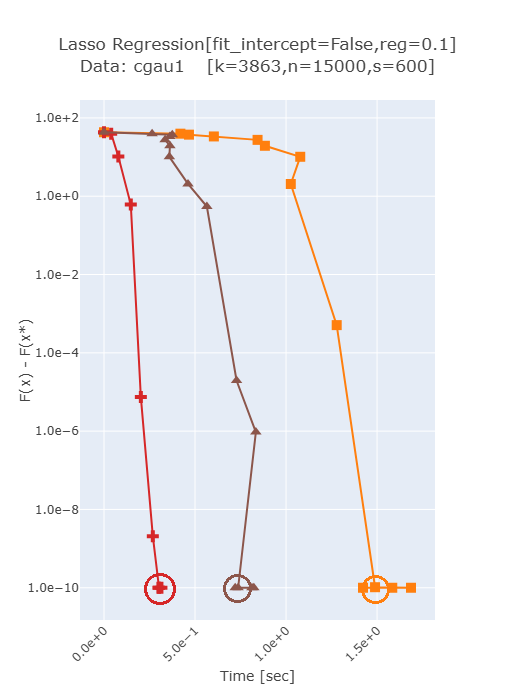}
    \includegraphics[width=.32\columnwidth]{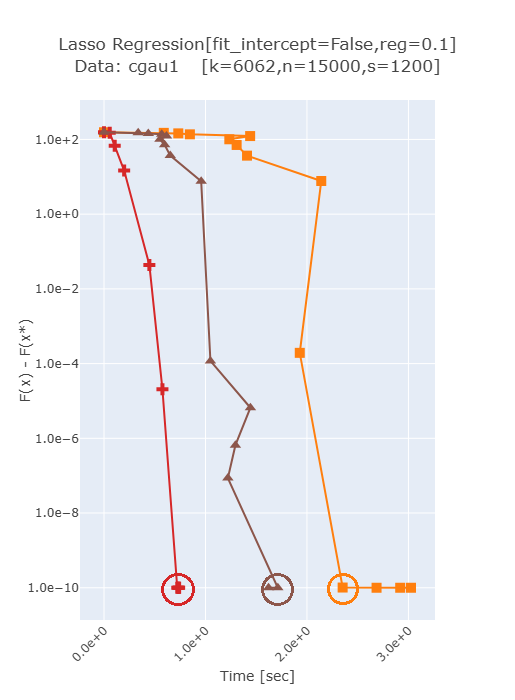}\\
    \includegraphics[width=.32\columnwidth]{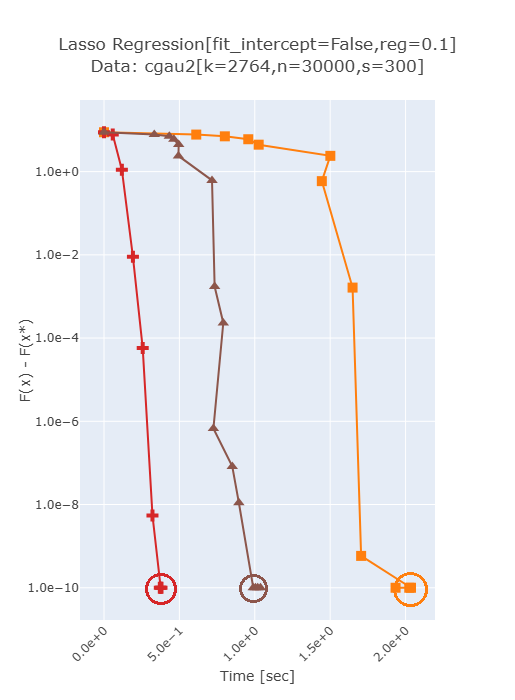}
    \includegraphics[width=.32\columnwidth]{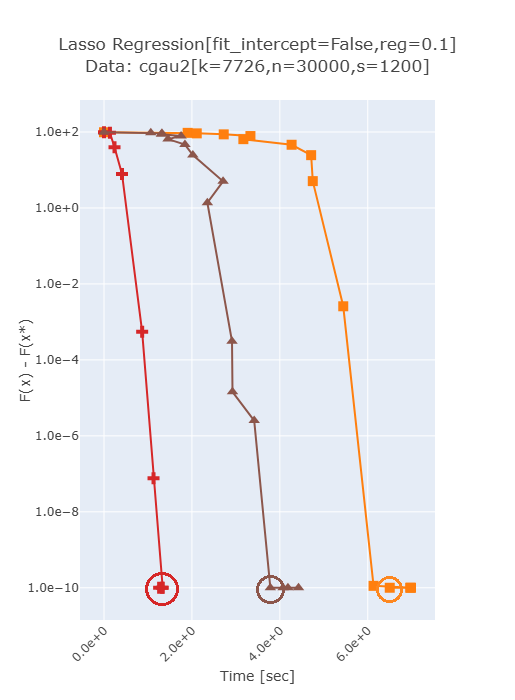}
    \includegraphics[width=.32\columnwidth]{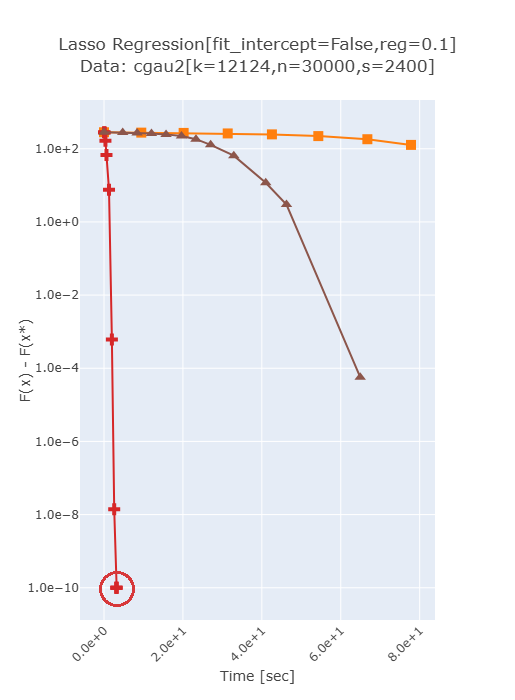}\\
    \caption{Plots of $F(\mat{x}_r) - F(\mat{x}_*)$ against running time.}
    \label{fig:sota_showcase_12}
\end{figure}

\begin{figure}
    \centering
    \includegraphics[width=0.35\columnwidth]{Styles/sota/Legend.png}\\
    \includegraphics[width=.32\columnwidth]{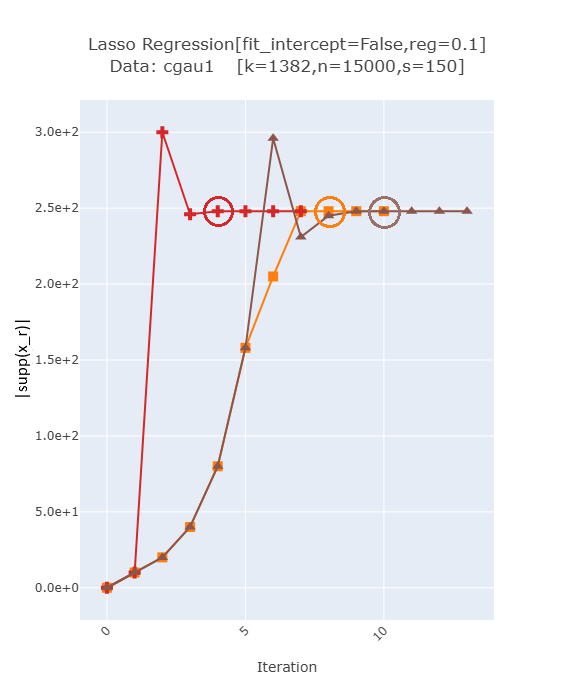}
    \includegraphics[width=.32\columnwidth]{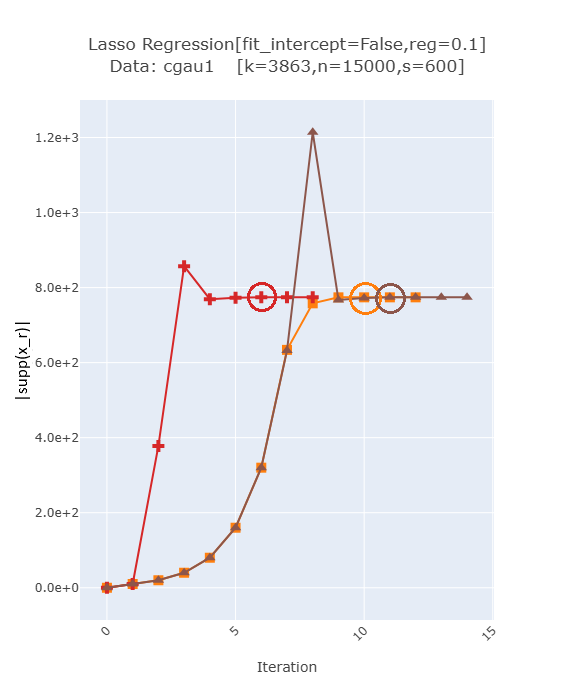}
    \includegraphics[width=.32\columnwidth]{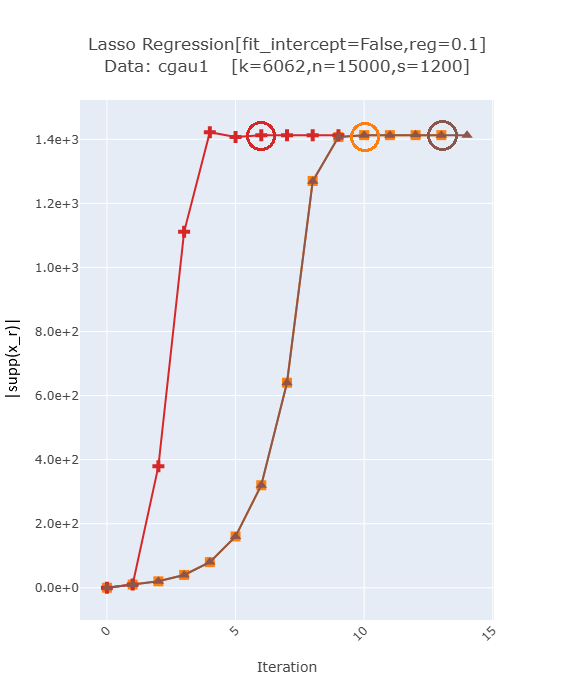}\\
    \includegraphics[width=.32\columnwidth]{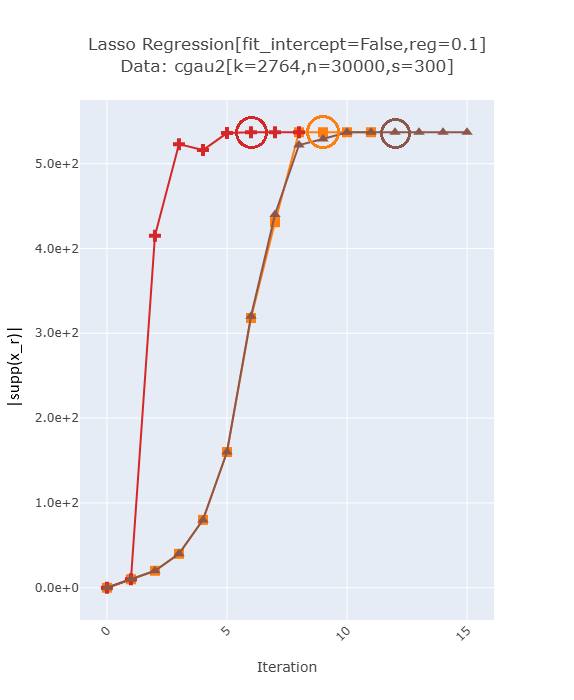}
    \includegraphics[width=.32\columnwidth]{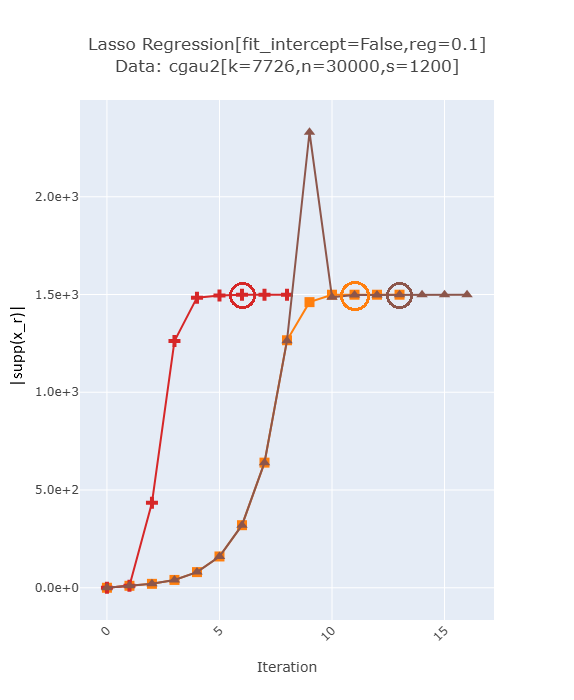}
    \includegraphics[width=.32\columnwidth]{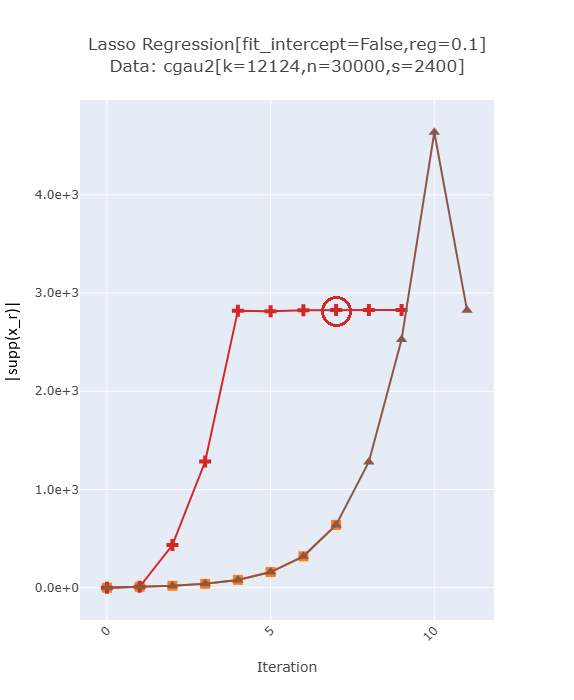}\\
    \includegraphics[width=.32\columnwidth]{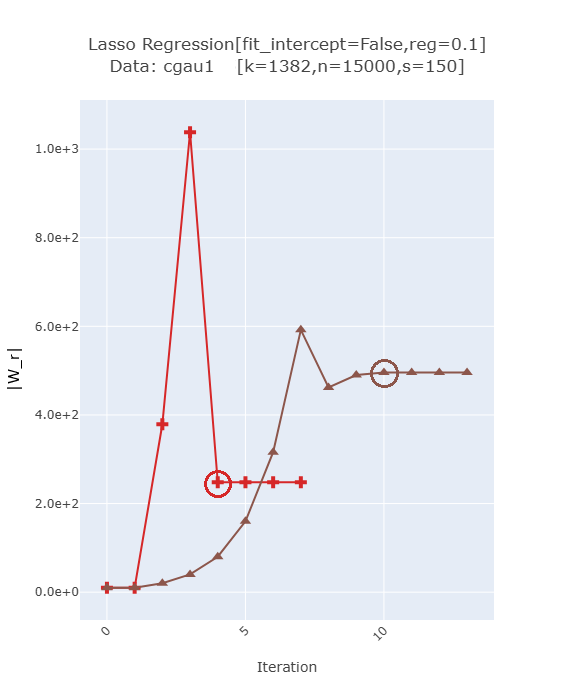}
    \includegraphics[width=.32\columnwidth]{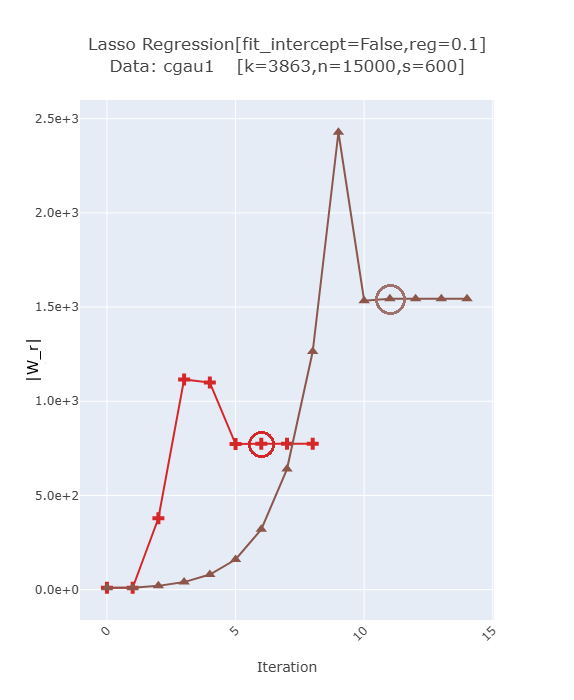}
    \includegraphics[width=.32\columnwidth]{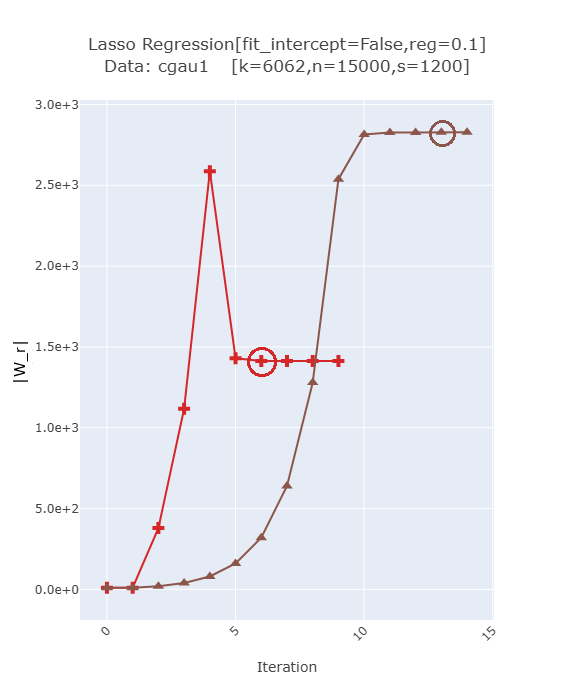}\\
    \includegraphics[width=.32\columnwidth]{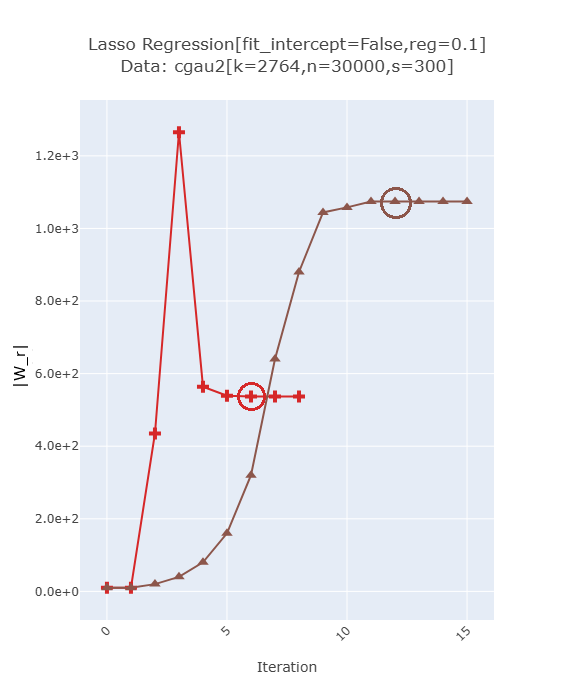}
    \includegraphics[width=.32\columnwidth]{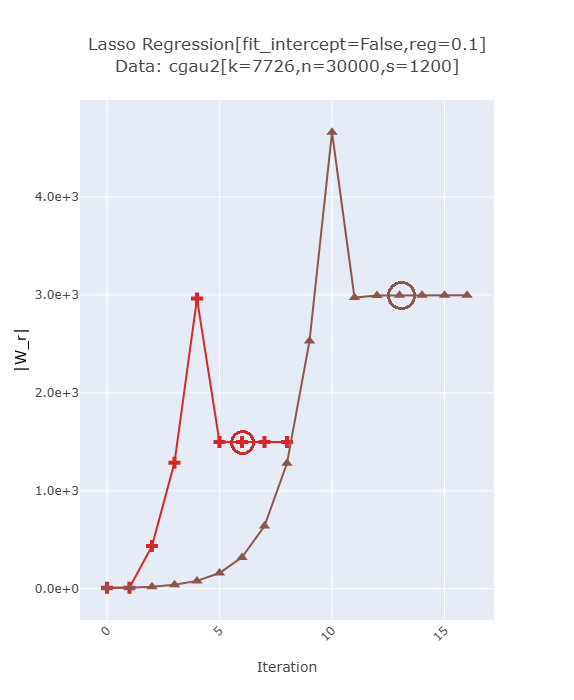}
    \includegraphics[width=.32\columnwidth]{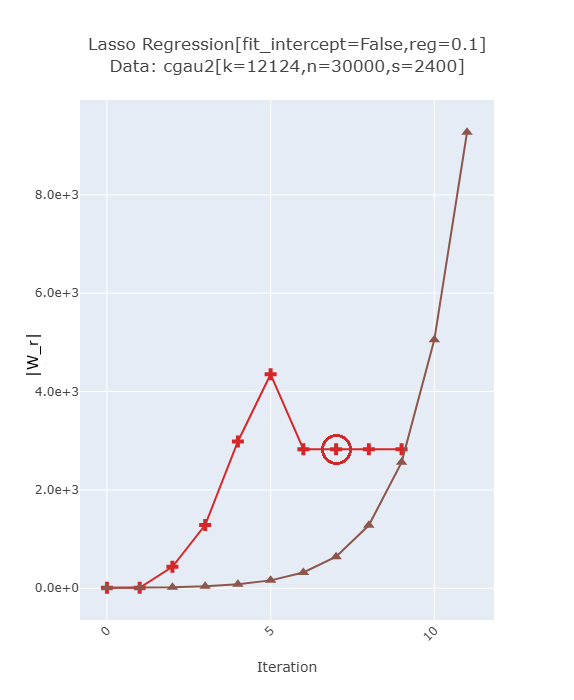}\\
    \caption{The top two rows show the plots of support set sizes against iteration.  The bottom two rows show the plots of working set sizes against iteration.}
    \label{fig:sota_supp_wss_12}
\end{figure}

We assume no knowledge of $s$.  As in Skglm and Celer~\cite{skglm,celer2}, DWS starts with a working set of size $p_0 = 10$ ($|E_0|$ is typically larger than 10).  Figure~\ref{fig:sota_showcase_12} shows the running times for some random inputs for $n \in \{15000,30000\}$. Skglm and Celer timed out in some runs; in those cases, no data point of their plots is circled (which indicates termination).  When Skglm and Celer did not time out, DWS is at least $1.91\times$ faster than Skglm and at least $3.0\times$ faster than Celer.  The top two rows in Figure~\ref{fig:sota_supp_wss_12} show the plots of the support set sizes.  The three methods give the same final support set size which is about 38\% larger than $s$ on average.
%

Since Skglm is more efficient than Celer, we will focus on comparing DWS with Skglm.  There are two main reasons for the speedup of DWS over Skglm.
Refer to the bottom two rows in Figure~\ref{fig:sota_supp_wss_12}. 
First, the working set size in DWS increases faster than in Skglm which yields a faster convergence.  Second, although the working set size in DWS may increase to much larger than the final support set size near the end of the computation, it is promptly reduced in the next iteration and kept smaller afterward.  In contrast, Skglm sustains a much larger working set (roughly twice as large) over multiple iterations near the end of the computation, which makes these iterations run significantly slower.  
Figures~\ref{fig:sota_showcase_34} and~\ref{fig:sota_supp_wss_34} 
show similar trends in the experimental results for some random inputs for $n \in \{45000,60000\}$.  

\begin{figure}[]
	\centering
	\includegraphics[width=0.35\columnwidth]{Styles/sota/Legend.png}\\
	\includegraphics[width=.32\columnwidth]{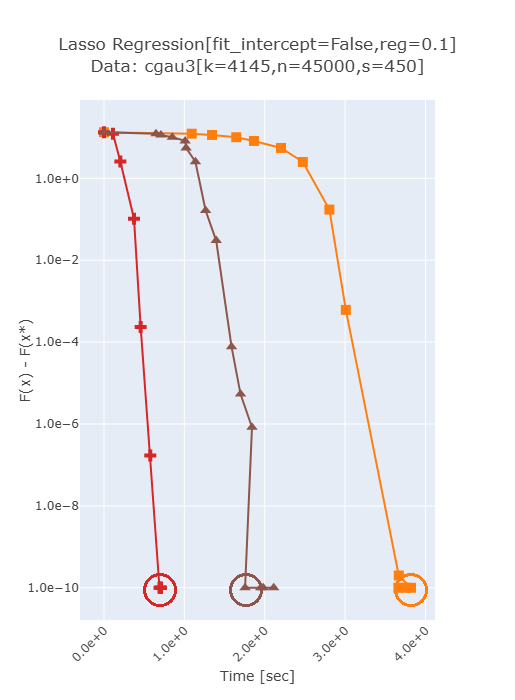}
	\includegraphics[width=.32\columnwidth]{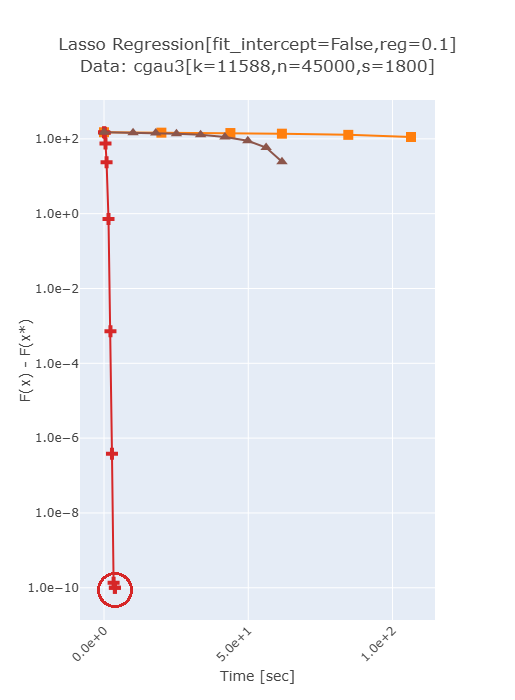}
	\includegraphics[width=.32\columnwidth]{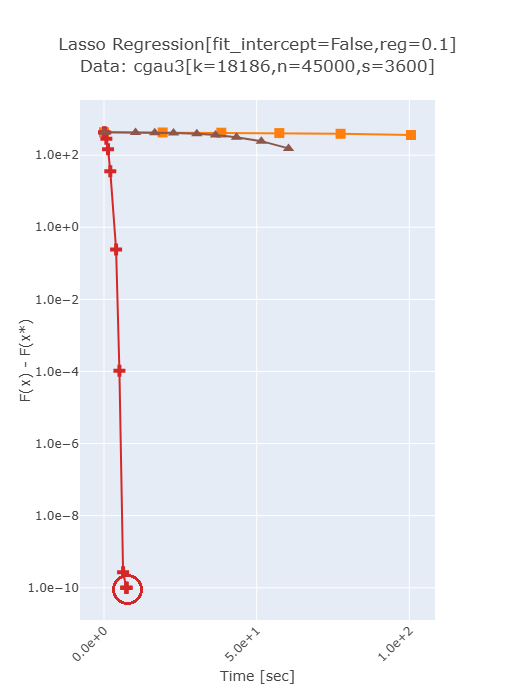}\\
	\includegraphics[width=.32\columnwidth]{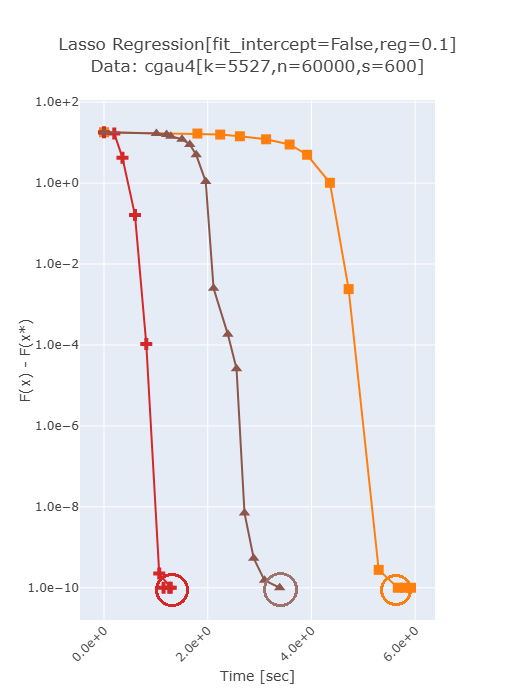}
	\includegraphics[width=.32\columnwidth]{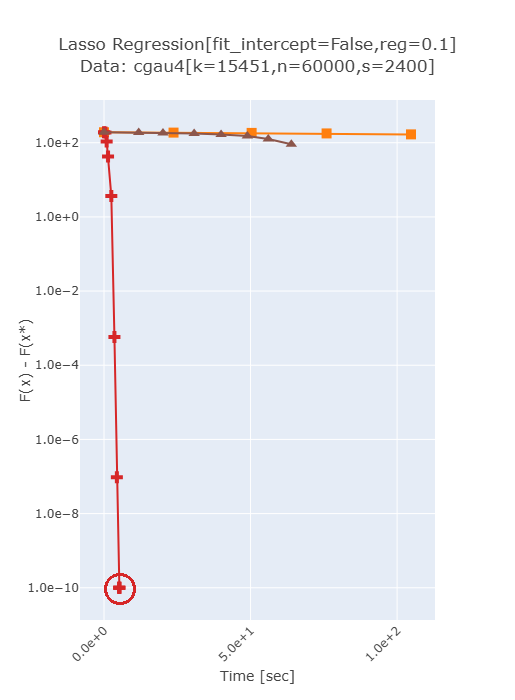}
	\includegraphics[width=.32\columnwidth]{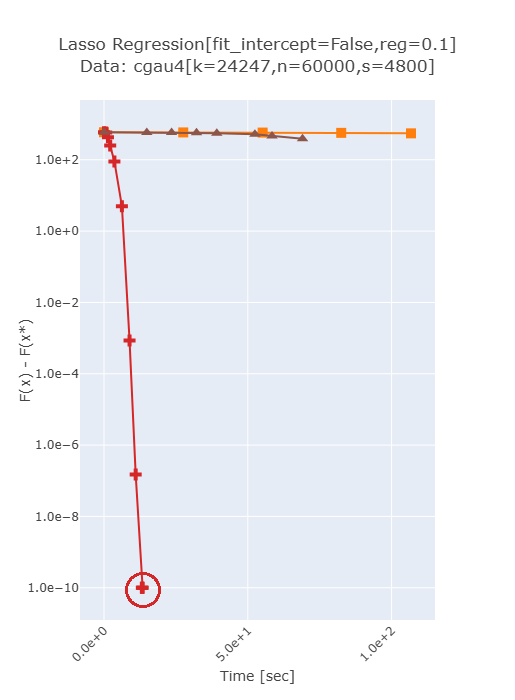}\\
	\caption{Plots of $F(\mat{x}_r) - F(\mat{x}_*)$ against running time for $n \in \{45000,60000\}$.}
	\label{fig:sota_showcase_34}
\end{figure}

\begin{figure}[]
	\centering
	\includegraphics[width=0.35\columnwidth]{Styles/sota/Legend.png}\\
	\includegraphics[width=.32\columnwidth]{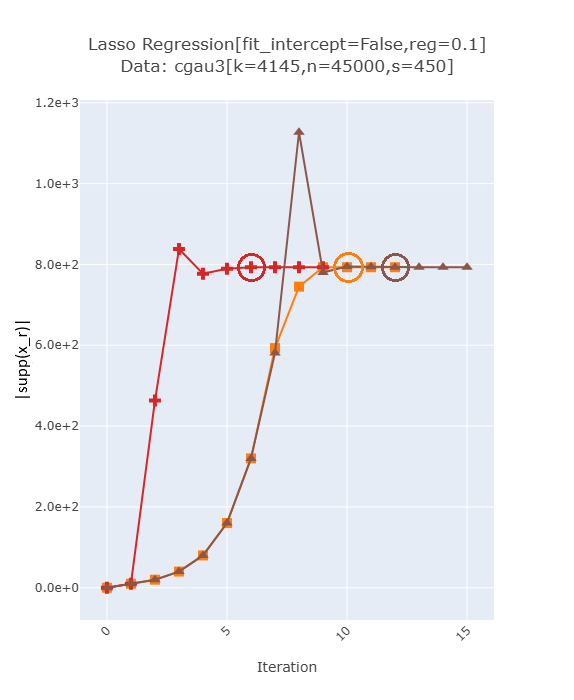}
	\includegraphics[width=.32\columnwidth]{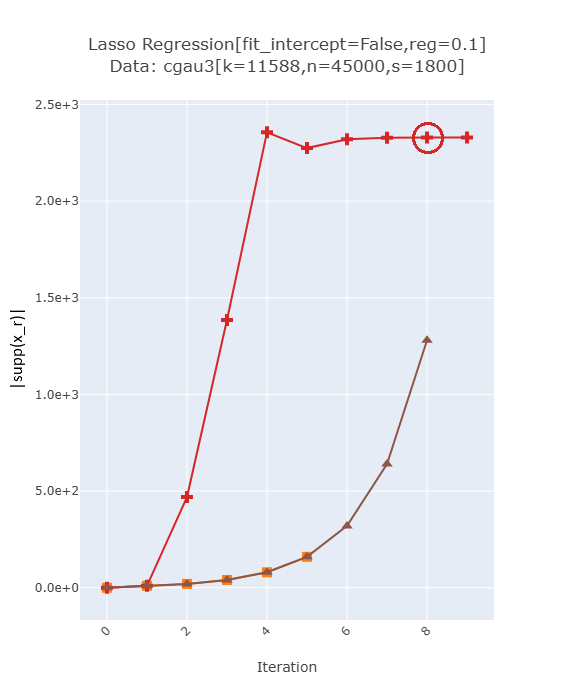}
	\includegraphics[width=.32\columnwidth]{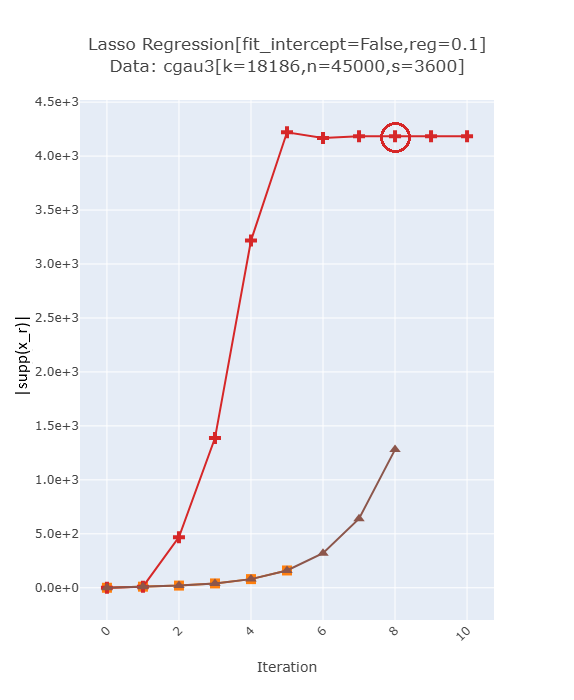}\\
	\includegraphics[width=.32\columnwidth]{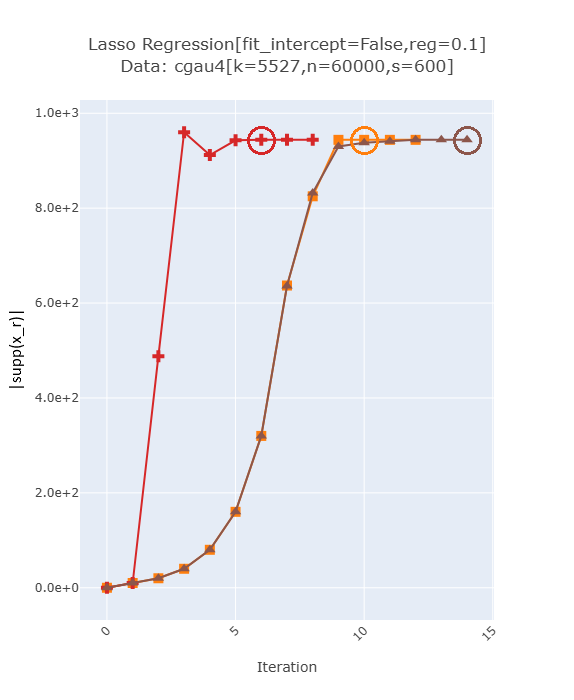}
	\includegraphics[width=.32\columnwidth]{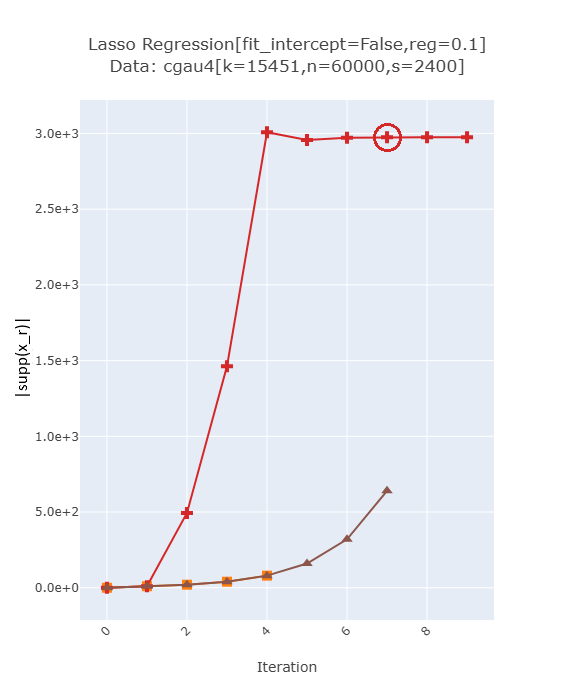}
	\includegraphics[width=.32\columnwidth]{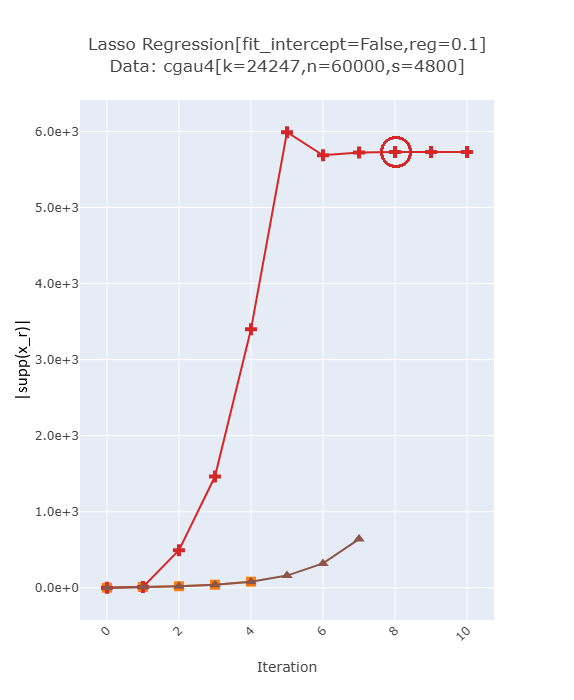}\\
	\includegraphics[width=.32\columnwidth]{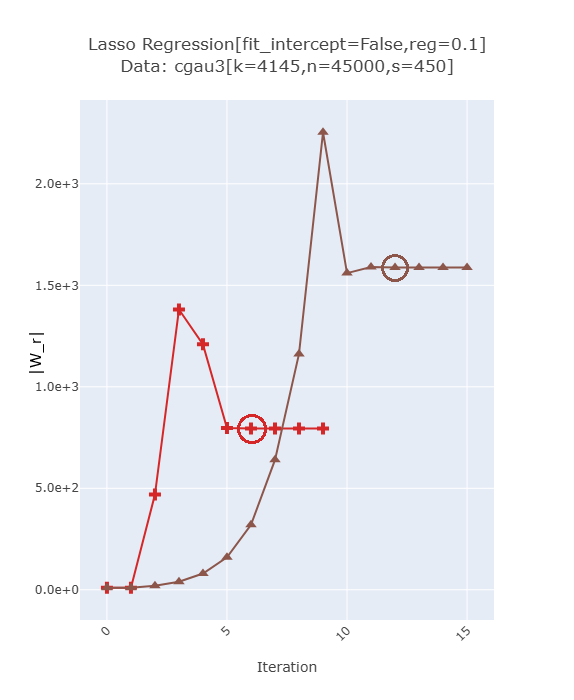}
	\includegraphics[width=.32\columnwidth]{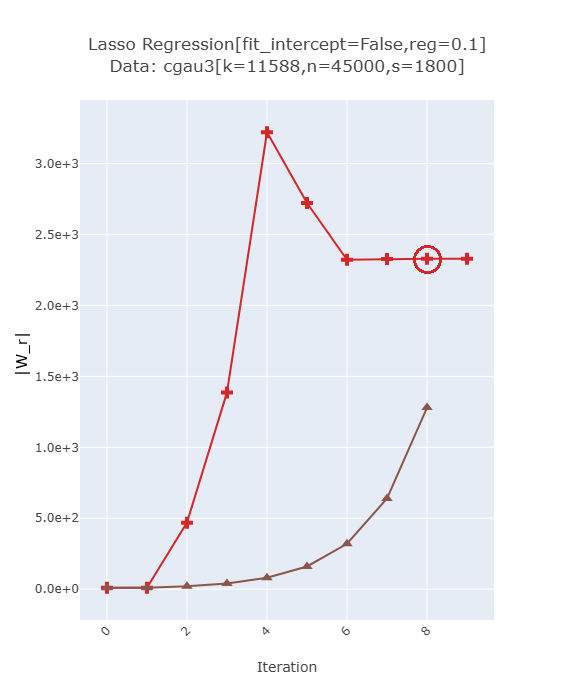}
	\includegraphics[width=.32\columnwidth]{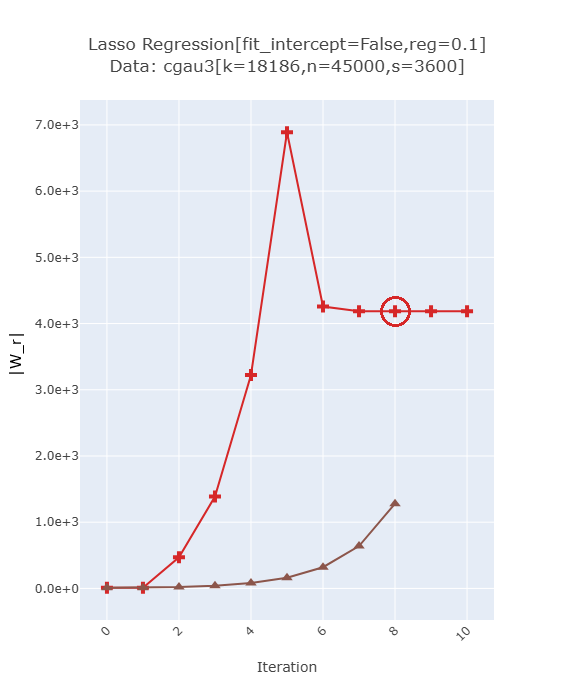}\\
	\includegraphics[width=.32\columnwidth]{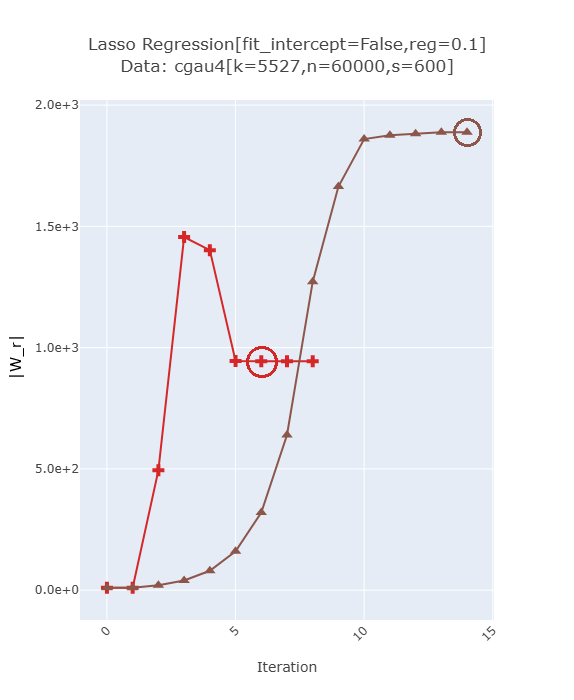}
	\includegraphics[width=.32\columnwidth]{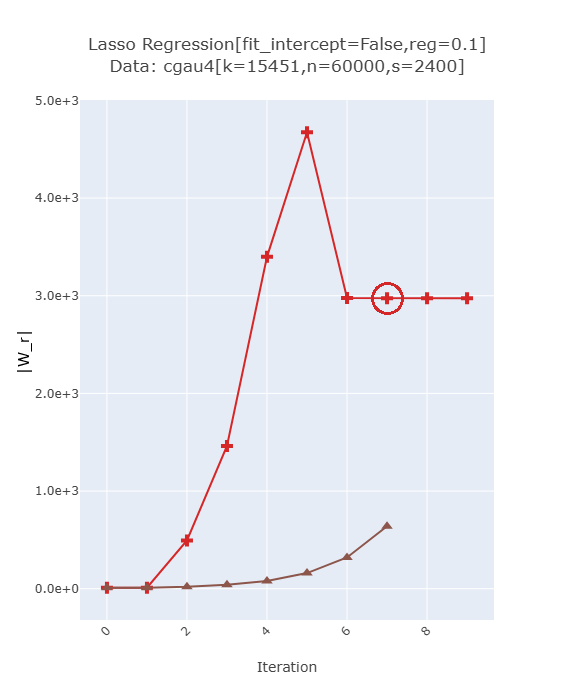}
	\includegraphics[width=.32\columnwidth]{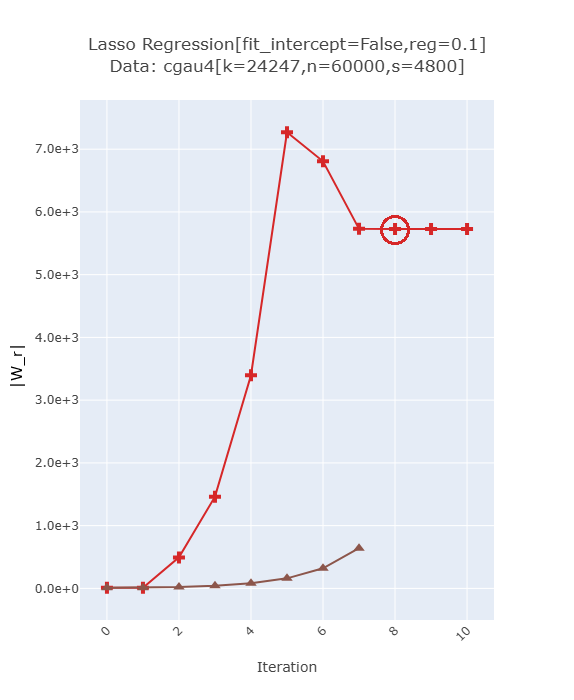}\\
	\caption{The top two rows show the plots of support set sizes against iteration for $n \in \{45000,60000\}$.  The bottom two rows show the plots of working set sizes against iteration for $n \in \{45000,60000\}$.}
	\label{fig:sota_supp_wss_34}
\end{figure}

\begin{figure}[]
	\centering
	\includegraphics[width=0.35\columnwidth]{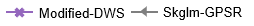}\\
	\includegraphics[width=.32\columnwidth]{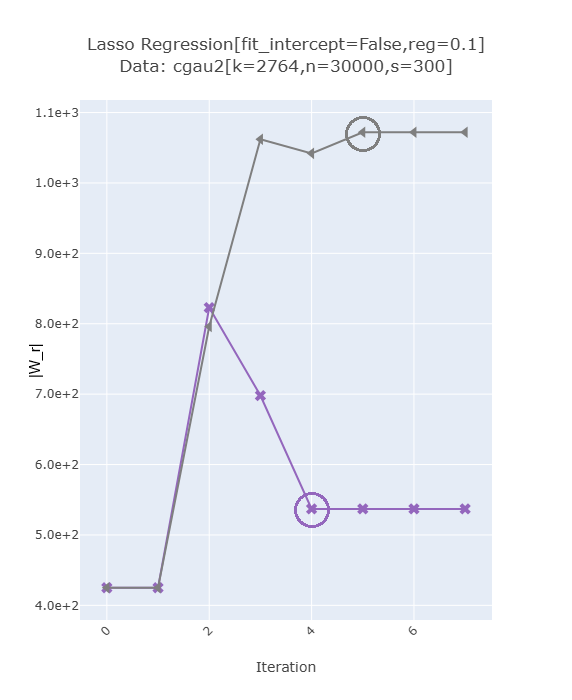}
	\includegraphics[width=.32\columnwidth]{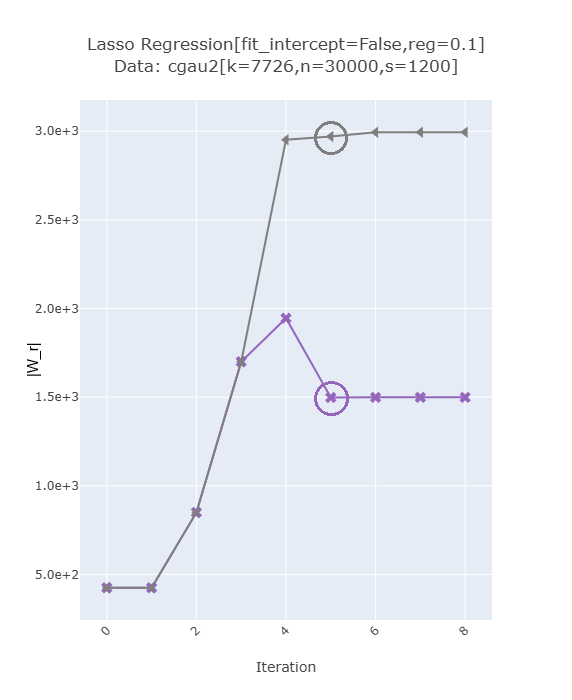}
	\includegraphics[width=.32\columnwidth]{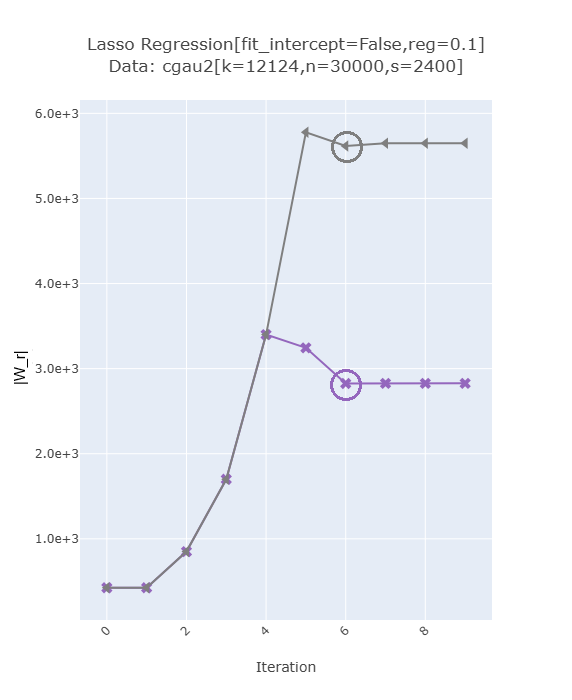}\\
	\includegraphics[width=.32\columnwidth]{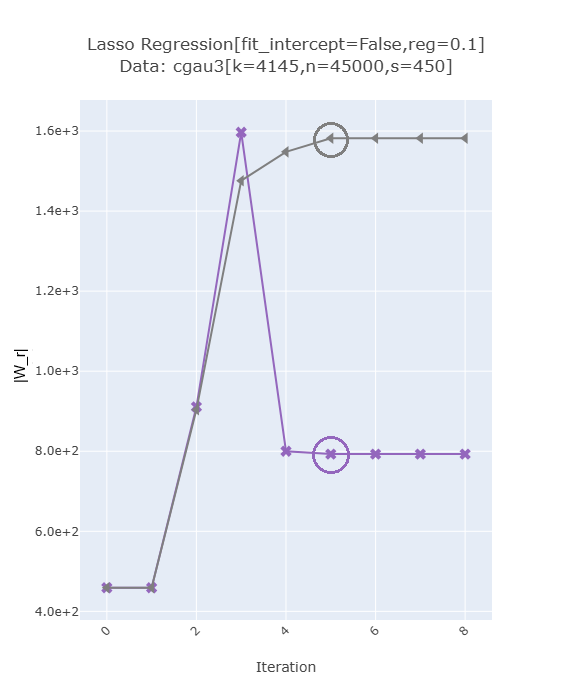}
	\includegraphics[width=.32\columnwidth]{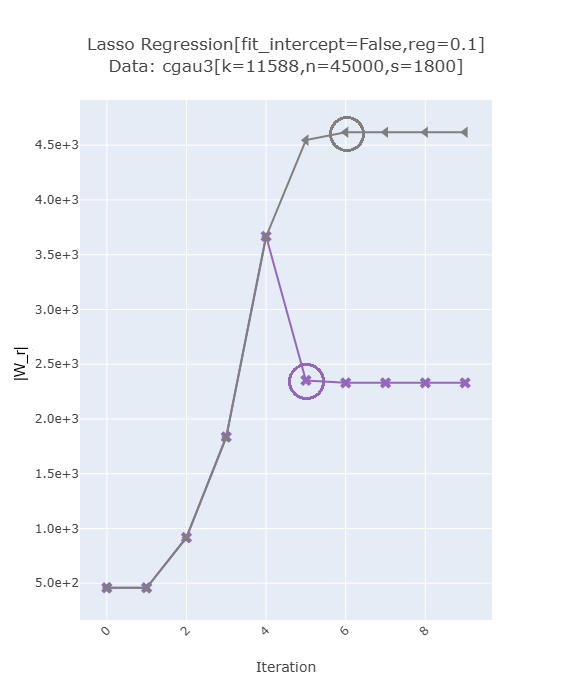}
	\includegraphics[width=.32\columnwidth]{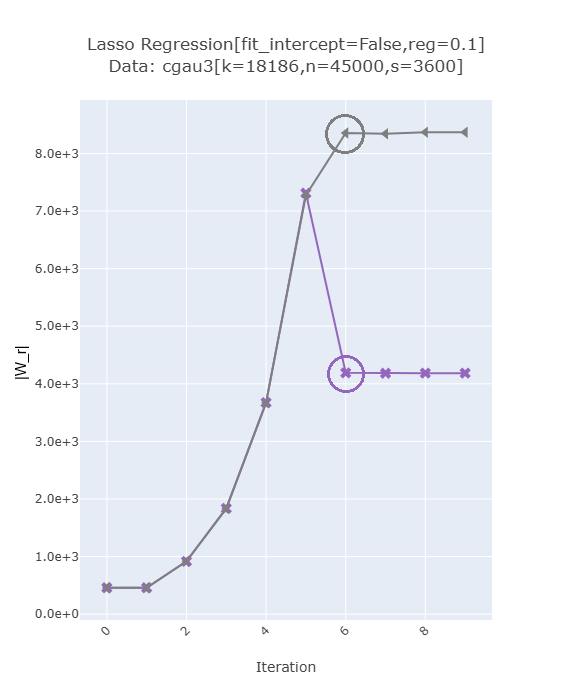}\\
	\includegraphics[width=.32\columnwidth]{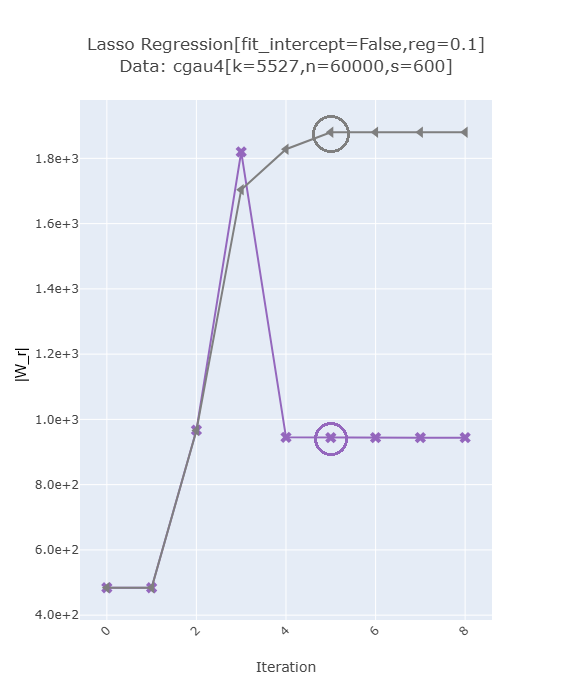}
	\includegraphics[width=.32\columnwidth]{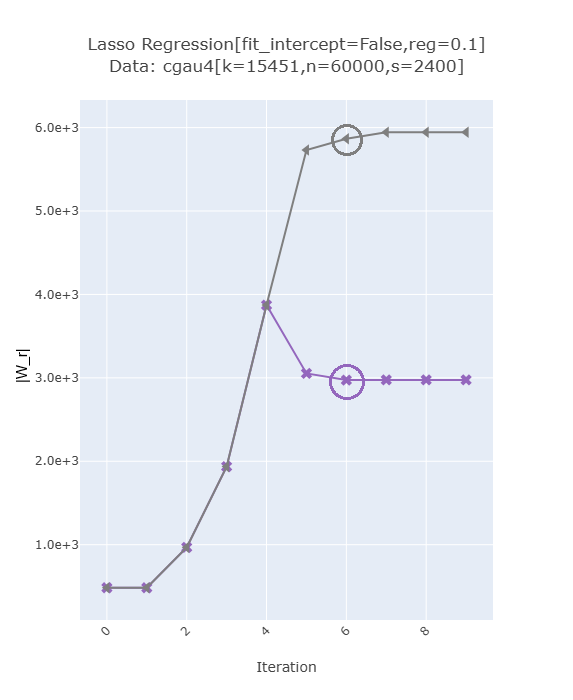}
	\includegraphics[width=.32\columnwidth]{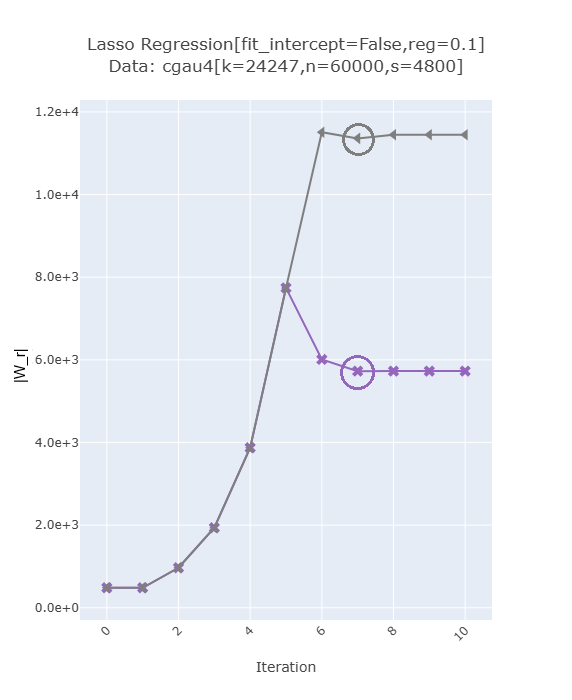}\\
	\caption{Plots of working set sizes against iteration. }
	\label{fig:nz0_wss_1234}
\end{figure}

\begin{figure}
	\centering
	\includegraphics[width=.32\columnwidth]{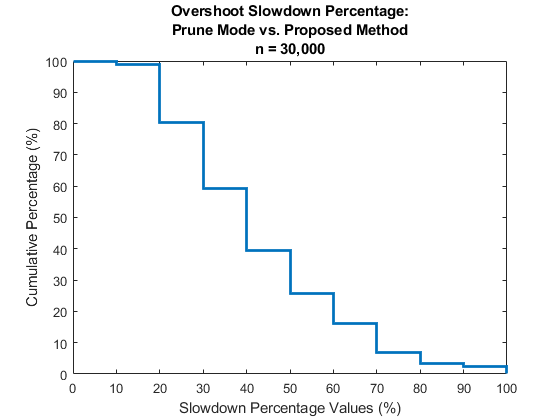}
	\includegraphics[width=.32\columnwidth]{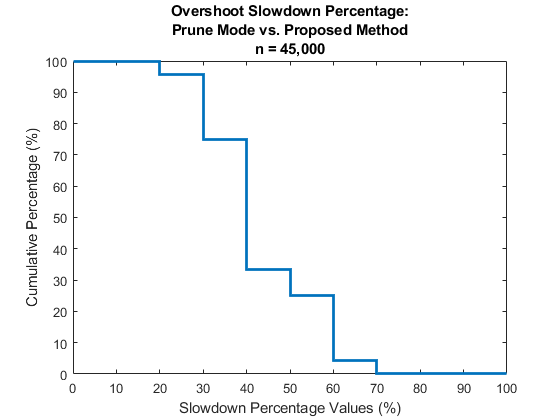}
	\includegraphics[width=.32\columnwidth]{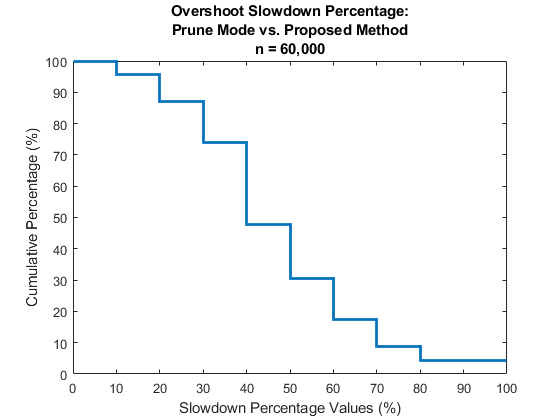}\\[4mm]
	\caption{
		The vertical axes show the cumulative percentages of test cases.  The horizontal axes show the slowdown percentage defined as
		$((\text{runtime of Skglm-GPSR} - \text{runtime of modified-method})/\text{runtime of modified-method}) \times 100\%$.  }
	\label{fig:hsovershoot}
\end{figure}

How important is the ability of DWS to scale back the working set size?  We study this question as follows.  First, we implemented the doubling method of Skglm with GPSR as the solver so that the comparison is on the same footing.  We refer to the resulting variant as Skglm-GPSR.  Second, we set $p_0 = \tau$ for both Skglm-GPSR and DWS, and we pretend that $|\supp(\mat{x}_0)| = \tau$ in DWS although $\mat{x}_0$ is still the zero vector.  We refer to the resulting variant as modified-DWS.  Figure~\ref{fig:nz0_wss_1234} shows that Skglm-GPSR and modified-DWS have a nearly common working set size (around $2^{r-1}\tau$ in the $r$-th iteration) until the computation is near the end.  
Therefore, there is no issue with the working set size increasing faster in modified-DWS or Skglm-GPSR.  Near the end of the computation, the working set size is scaled back in modified-DWS, whereas the working set size in Skglm-GPSR is roughly twice as large.  We tried 269 cases for $n=30000$, 74 cases for $n = 45000$, and 74 cases for $n = 60000$. 
Refer to Figure~\ref{fig:hsovershoot}.
Skglm-GPSR is slower by 20\% or more in at least 80\% of the cases, by 30\% or more in at least 59\% of the cases, and by 40\% or more in at  least 30\% of the cases.  The ability of DWS to scale back the working set size improves efficiency significantly.

\section{Theoretical analysis}
\label{sec:analysis}

 Given any function $\varphi: \real^n \rightarrow \real$, a vector $\xi \in \mathbb{R}^n$ such that $\varphi(\mat y) \ge \varphi(\mat x) + \langle \xi, \mat y - \mat x \rangle$ for all $\mat y \in \mathbb{R}^\nu$ is called a \emph{subgradient} of $\varphi$ at $\mat{x}$~\cite{nonl}.  For a smooth function, the subgradient at a point is unique and equal to the gradient, which is denoted by $\nabla \varphi$.  There are multiple subgradients at a non-smooth point $\mat{x}$; we use $\partial \varphi(\mat{x})$ to denote the set of all subgradients of $\varphi$ at $\mat{x}$.  

We have $\partial F(\mat{x}) = \{ \nabla f(\mat{x}) + \xi : \xi \in \partial g(\mat{x})\}$.  A vector $\mat{n}$ is a \emph{descent direction} from $\mat{x}$ if and only if $\sup\{ \langle \gamma, \mat{n} \rangle : \gamma \in \partial F(\mat{x}) \} < 0$.  
The function $F$ is minimized at $\mat{x}$ if and only if $\partial F(\mat{x})$ contains the zero vector~\cite{nonl}.

Lemma~\ref{lem:32} proves the termination condition of $E_r = \emptyset$.  Theorem~\ref{thm:tau-sum} analyzes the sizes of the working and support sets.  

\newpage

\begin{lemma}
\label{g_subgradient-t}
\hspace*{10pt}
\begin{enumerate}[{\em (i)}]
\item Take any $\mat{x} \in \real^n$.
\begin{enumerate}[{\em (a)}]
\item $\forall\, \xi \in \partial g(\mat x), \, \forall\, i \in [n]$, if $(\mat x)_i = 0$, then $(\xi)_i \in [-\eta,\eta]$; otherwise, $(\xi)_i = \mathrm{sign}((\mat{x})_i) \cdot \eta$.

\item Every vector that satisfies the conditions in {\em (i)(a)} is a subgradient in $\partial g(\mat x)$.

\end{enumerate}

\item $\forall\, i \in E_r, \, \forall\, \gamma \in \partial F(\mat{x}_r)$, $\mathrm{sign}((\gamma)_i) = \mathrm{sign}((\nabla f(\mat{x}_r))_i) \in \{-1,1\}$.

\item $E_r = \bigl\{i \in [n] : \forall\, \gamma \in \partial F(\mat{x}_r), (\gamma)_i \not= 0 \bigr\}$.

\end{enumerate}
\end{lemma}
\begin{proof}
	Consider (i)(a).  Take any $\xi \in \partial g(\mat{x})$.
	By the definition of a subgradient, for all $\mat y \in\real^n$, $g(\mat y)-g(\mat x) \ge   \langle \xi, \mat y - \mat x \rangle$, which is equivalent to 
	\begin{equation}
		\eta \sum_{i=1}^n |(\mat y)_i| - \eta \sum_{i=1}^n |(\mat x)_i| \ge \sum_{i=1}^n  (\xi)_i \cdot (\mat y - \mat x)_i.
		\label{eq:g-1}
	\end{equation}
	
	Take any index $i \in [n]$.  Let $Y_i = \bigl\{ \mat{y} \in \real^n : \forall j \not= i , (\mat y)_j=(\mat x)_j \bigr\}$.  Clearly, for every $\mat{y} \in Y_i$, applying \eqref{eq:g-1} to $\mat{y}$ gives 
	\[
	\eta |(\mat y)_i| - \eta |(\mat x)_i| \ge (\xi)_i \cdot (\mat y - \mat x)_i.
	\]
	
	\begin{itemize}
		
		\item Case 1:~Suppose that $(\mat x)_i \geq 0$.  Choose $\mat{y} \in Y_i$ such that $(\mat y)_i > (\mat x)_i$. Then, $\eta(\mat y - \mat x)_i = \eta|(\mat y)_i|-\eta|(\mat x)_i| \ge (\xi)_i \cdot  (\mat y - \mat x)_i$.  Dividing both sides by $(\mat y- \mat x)_i$ gives $(\xi)_i\le \eta$. 
		
		\item Case 2:~Suppose that $(\mat x)_i \leq 0$.  Choose $\mat{y} \in Y_i$ such that $(\mat y)_i < (\mat x)_i$.  Then, $-\eta(\mat y - \mat x)_i =\eta|(\mat y)_i|-\eta|(\mat x)_i| \ge (\xi)_i \cdot  (\mat y - \mat x)_i$.  Dividing both sides by $(\mat y- \mat x)_i$ gives $(\xi)_i\ge -\eta$. 
		
	\end{itemize}
	
	Combining cases~1 and~2 gives $(\xi)_i \in [-\eta,\eta]$ when $(\mat{x})_i = 0$.  Suppose that $(\mat{x})_i > 0$.  We already have $(\xi)_i \leq \eta$ by case~1. Choose $\mat{y} \in Y_i$ such that $0 < (\mat{y})_i < (\mat x)_i$.  Then, $\eta(\mat y - \mat x)_i = \eta|(\mat y)_i|-\eta|(\mat x)_i|\ge (\xi)_i \cdot  (\mat y - \mat x)_i$.  Dividing both sides by $(\mat{y}-\mat{x})_i$ gives $(\xi)_i \geq \eta$.  As a result, $(\xi)_i= \eta$.  Suppose that $(\mat x)_i < 0$.  We already have $(\xi)_i \geq -\eta$ by case~2.  Choose $\mat{y} \in Y_i$ such that $(\mat{x}_r)_i < (\mat{y})_i < 0$.  Then, $-\eta(\mat{y} - \mat{x}_r)_i = \eta|(\mat y)_i|-\eta|(\mat x)_i| \ge (\xi)_i \cdot  (\mat y - \mat x)_i$.  Dividing both sides by $(\mat y- \mat x)_i$  gives $(\xi)_i \leq -\eta$.  In all, $(\xi)_i= -\eta$.  This completes the proof of (i)(a).
	
	Consider (i)(b).  Take any vector $\xi \in \real^n$ that satisfies the conditions in (i)(a).  Under these conditions, it is easy to verify that for every $\mat{y} \in \real^n$ and every $i \in [n]$, $\eta |(\mat y)_i| - \eta |(\mat x)_i| \ge (\xi)_i \cdot (\mat y - \mat x)_i$.  Then, for all $\mat{y} \in \real^n$, 
	\begin{align*}
		g(\mat{y}) - g(\mat{x}) = \eta \sum_{i=1}^n |(\mat y)_i| - \eta \sum_{i=1}^n |(\mat x)_i| \ge \sum_{i=1}^n (\xi)_i \cdot (\mat y - \mat x)_i = \langle \xi, \mat y - \mat x \rangle,
	\end{align*}
	which implies that $\xi$ is a subgradient in $\partial g(\mat{x})$.
	
	Consider (ii).  Take any $i \in E_r$.  By definition, $|(\nabla f(\mat{x}_r))_i| > \eta$. So $(\nabla f(\mat{x}_r))_i \not= 0$. It also follows from (i) that for every $\xi \in \partial g(\mat{x}_r)$, if $(\nabla f(\mat{x}_r))_i > 0$, then $(\nabla f(\mat{x}_r))_i + (\xi)_i > 0$, and if $(\nabla f(\mat{x}_r))_i < 0$, then $(\nabla f(\mat{x}_r))_i + (\xi)_i < 0$.  Note that $\partial F(\mat{x}_r) = \{ \nabla f(\mat{x}_r) + \xi : \xi \in \partial g(\mat{x}_r)\}$.  In other words, for every $i \in E_r$ and every $\gamma \in \partial F(\mat{x}_r)$, $\mathrm{sgn}((\gamma)_i) = \mathrm{sgn}((\nabla f(\mat{x}))_i) \in \{-1,1\}$.  This proves (ii).
	
	Consider (iii).   For all $i \not\in [n] \setminus (W_r \cup E_r)$, we have $(\mat{x}_r)_i = 0$ and $|(\nabla f(\mat{x}_r))_i| \leq \eta$.  By (i)(b), for all $i \not\in [n] \setminus (W_r \cup E_r)$, every value in $[-\eta,\eta]$ is a legitimate $i$-th coordinate for a subgradient of $g$ at $\mat{x}_r$, which includes $-(\nabla f(\mat{x}_r))_i$.  Hence, there exists $\gamma \in \partial F(\mat{x}_r)$ such that $(\gamma)_i = 0$ for all $i \not\in [n] \setminus (W_r \cup E_r)$.  Since $\mat{x}_r$ is the optimal solution with respect to the working set $W_r$, there exists $\gamma \in \partial F(\mat{x}_r)$ such that $(\gamma)_i = 0$ for all $i \in W_r$.  We conclude that for every $i \in [n] \setminus E_r$, there exists $\gamma \in \partial F(\mat{x}_r)$ such that $(\gamma)_i = 0$.  By the result in (ii), for every $i \in E_r$ and every $\gamma \in \partial F(\mat{x}_r)$, $(\gamma)_i \not = 0$.  This proves the correctness of (iii), i.e., $i \in E_r$ if and only if $(\gamma)_i \not= 0$ for all $\gamma \in \partial F(\mat{x}_r)$.\hfill\qed
\end{proof}

\begin{lemma}
\label{lem:32}
Let $\mat{e}_i$ be the unit vector in the direction of the positive $i$-th axis.  Every unit conical combination of $\{-\mathrm{sign}((\nabla f(\mat x_r))_i) \cdot \mat{e}_i : i \in E_r\}$ is a descent direction from $\mat{x}_r$.  If $E_r = \emptyset$, then $\mat{x}_r$ is the global minimum.
\end{lemma}
\begin{proof}
	Let $\rho = \min_{i \in E_r} |\nabla f(\mat{x}_r)_i| - \eta$ which is positive.  Take any $i \in E_r$ and any $\gamma \in \partial F(\mat{x}_r)$.  By Lemma~\ref{g_subgradient-t}(i), the $i$-th coordinate of any subgradient in $\partial g(\mat{x}_r)$ is in the range $[-\eta,\eta]$, which implies that $|(\gamma)_i| \geq \rho$.  Let $\mat{s}_i = -\mathrm{sign}((\nabla f(\mat{x}_r))_i) \cdot \mat{e}_i$.  By Lemma~\ref{g_subgradient-t}(ii), $\langle \gamma,\mat{s}_i \rangle = -|(\gamma)_i| \leq -\rho < 0$.  For every unit conical combination $\sum_{i \in E_r} \alpha_i \mat{s}_i$, some coefficient $\alpha_i$ is at least $1/\sqrt{n}$.  Thus, $\sup_{\gamma \in \partial F(\mat{x}_r)} \langle \gamma,\sum_{i \in E_r} \alpha_i \mat{s}_i\rangle \leq -\rho/\sqrt{n} < 0$, proving that $\sum_{i \in E_r} \alpha_i\mat{s}_i$ is a descent direction.  If $E_r$ is empty, by Lemma~\ref{g_subgradient-t}(iii), for every $i \in [n]$, there exists $\xi \in \partial g(\mat{x}_r)$ such that $(\xi)_i = -(\nabla f(\mat{x}_r))_i$, that is, $-(\nabla f(\mat{x}_r))_i$ is a legitimate $i$-th coordinate of a subgradient in $\partial g(\mat{x}_r)$.  It follows from Lemma~\ref{g_subgradient-t}(i)(b) that the zero vector belongs to $\partial F(\mat{x}_r)$, which implies that $\mat{x}_r$ is the global minimum.\hfill\qed
\end{proof}

Define a vector $\zeta_{\mat{x}_r} \in \real^n$ such that for all $i \in [n]$, if $i \not\in E_r$, then $(\zeta_{\mat x_r})_i=-(\nabla f(\mat x_r))_i$, and if $i \in E_r$, then $(\zeta_{\mat x_r})_i= - \mathrm{sign}((\nabla f(\mat x_r))_i)\cdot \eta$. Define $\gamma_{\mat x_r} = \nabla f(\mat x_r) + \zeta_{\mat x_r}$.   
Let $\mat{x}_*$ denote the optimal solution that minimizes $F$.  Given a vector $\mat{v}$ and a subset $S \subseteq [n]$, $\mat{v}  \downarrow  S$ denotes the orthogonal projection of $\mat{v}$ in the subspace spanned by $\{\mat{e}_i : i \in S\}$.

\begin{lemma}
\label{lem:special}
$\zeta_{\mat{x}_r} \in \partial g(\mat{x}_r)$ and $\gamma_{\mat{x}_r} \in \partial F(\mat{x}_r)$.
\end{lemma}
\begin{proof}
	By Lemma~\ref{g_subgradient-t}(ii), for all $i \in E_r$, $(\zeta_{\mat x_r})_i \in \{-\eta,\eta\}$.   For all $i \in E_r$, $(\mat{x}_r)_i = 0$ as $E_r \cap W_r = \emptyset$. By Lemma~\ref{g_subgradient-t}(i)(b), for all $i \in E_r$, $(\zeta_{\mat{x}_r})_i$ is a legitimate $i$-th coordinate for a subgradient in $\partial g(\mat{x}_r)$.  By Lemma~\ref{g_subgradient-t}(iii), for all $i \not\in E_r$, there exists $\gamma \in \partial F(\mat{x}_r)$ such that $(\gamma)_i = 0$. It means that for all $i \not\in E_r$, $(\zeta_{\mat{x}_r})_i = -(\nabla f(\mat x_r))_i$ must a legitimate $i$-th coordinate for a subgradient in $\partial g(\mat x_r)$ so that it cancels $(\nabla f(\mat x_r))_i$.   Thus, $\zeta_{\mat{x}_r} \in \partial g(\mat{x}_r)$.  Then, $\gamma_{\mat{x}_r} \in \partial F(\mat{x}_r)$ by definition.\hfill\qed
\end{proof}

\begin{lemma}
\label{gfy_to_gamma_y-t}
Let $\mat{n}_r$ be any unit conical combination of $\{-\mathrm{sign}((\nabla f(\mat x_r))_i) \cdot \mat{e}_i : i \in E_r\}$.  Let $\mat{y}_r$ be the point in direction $\mat{n}_r$ from $\mat{x}_r$ that minimizes $F$. 
\begin{enumerate}[{\em (i)}]

\item There exists $\xi \in \partial g(\mat y_r)$ such that  $ \langle  \mat \nabla f(\mat y_r) + \xi, \mat y_r - \mat x_r\rangle  = 0$.

\item For every $\xi \in \partial g(\mat{y}_r)$, both $\langle \zeta_{\mat{x}_r} - \xi, \mat{x}_r \rangle$ and $\langle \zeta_{\mat{x}_r} - \xi, \mat{y}_r\rangle$ are zero.

\item For every $\mat z \in \real^n$, $F(\mat x_r)-F(\mat z) \le - \langle  \gamma_{\mat x_r}, \mat  z - \mat x_r\rangle$.

\item $F(\mat x_r)-F(\mat y_r) = - \frac{1}{2}\langle \gamma_{\mat x_r}, \mat y_r - \mat x_r\rangle$.
\end{enumerate}
\end{lemma}
\begin{proof}
	Since $\mat{n}_r$ is a descent direction by Lemma~\ref{lem:32}, the point $\mat{y}_r$ is well defined. Consider (i).  Let $L$ denote the line through $\mat{x}_r$ parallel to $\mat{n}_r$.  Since the minimum of $F$ in $L$ is achieved at $\mat{y}_r$, it is known that there exists $\gamma \in \partial F(\mat y_r)$ such that $\langle \gamma, \mat{z} - \mat{y}_r \rangle\ge 0$ for all $\mat z \in L$.  Note that $\mat{x}_r \in L$.  Choose the point $\mat{x} \in L$ such that $\mat{x}_r - \mat{y}_r = \mat{y}_r - \mat{x}$.  Then, we have $\langle \gamma, \mat{x}_r - \mat{y}_r \rangle \geq 0$ and $\langle \gamma, \mat{x} - \mat{y}_r \rangle \geq 0$.  The second inequality also implies that $\langle \gamma, \mat{x}_r - \mat{y}_r \rangle = -\langle \gamma, \mat{x}-\mat{y}_r\rangle \leq 0$.  It follows that $\langle \gamma, \mat{x}_r - \mat{y}_r \rangle = 0$, which implies that there exists $\xi \in \partial g(\mat{y}_r)$ such that $\langle \nabla f(\mat{y}_r) + \xi, \mat{x}_r - \mat{y}_r \rangle = 0$.  This proves (i).
	
	Consider (ii).  Take any $i \in \supp(\mat{n}_r)$.  As $i \in E_r$ by definition, we have $(\mat{x}_r)_i = 0$ and $(\zeta_{\mat{x}_r})_i = -\mathrm{sign}((\nabla f(\mat{x}_r))_i) \cdot \eta$.  Also, $\mathrm{sign}((\mat{y}_r)_i) = -\mathrm{sign}((\nabla f(\mat{x}_r)_i)$ because we descend from $\mat{x}_r$ in direction $\mat{n}_r$ to reach $\mat{y}_r$.  By Lemma~\ref{g_subgradient-t}(i)(a), $(\xi)_i = \mathrm{sign}((\mat{y}_r)_i) \cdot \eta = -\mathrm{sign}((\nabla f(\mat{x}_r)_i) \cdot \eta$.  Therefore, $(\zeta_{\mat{x}_r})_i = (\xi)_i$.  For any $i \not\in \supp(\mat{n}_r)$, we have $(\mat{x}_r)_i = (\mat{y}_r)_i$.  If they are not zero, then $(\zeta_{\mat{x}_r})_i$ and $(\xi)_i$ are identical by Lemma~\ref{g_subgradient-t}(i)(a).  We conclude that both $((\zeta_{\mat{x}_r})_i - (\xi)_i) \cdot (\mat{x}_r)_i$ and  $((\zeta_{\mat{x}_r})_i - (\xi)_i) \cdot (\mat{y}_r)_i$ are zero for all $i \in [n]$.  This proves (ii).
	
	Before proving (iii) and (iv), we first prove the following equation: $\forall \, \mat x, \mat z \in \mathbb{R}^n, \forall\, \zeta_1 \in \partial g(\mat x), \forall \, \zeta_2 \in \partial g(\mat z)$,
	\begin{equation}
		F(\mat x)-F(\mat z) = -\frac{1}{2}\|\mat A (\mat z - \mat x)\|^2- \langle  \mat \nabla f(\mat x) + \zeta_1, \mat z - \mat x\rangle +\langle  \zeta_1 - \zeta_2, \mat z \rangle.
		\label{eq:-1}
	\end{equation}
	Take any $i\in [n]$ and any $\zeta_1 \in \partial g(\mat{x})$. By Lemma~\ref{g_subgradient-t}(i)(a), $(\zeta_1)_i \cdot (\mat x)_i= \eta|(\mat x)_i|$.  Therefore, $\langle {\zeta}_1  , \mat x \rangle = \sum_{i=1}^n (\zeta_1)_i \cdot (\mat x)_i = \eta \|\mat x\|_1 = g(\mat x)$, which implies that
	\begin{equation}
		g(\mat x)-g(\mat z)= \langle {\zeta}_1  , \mat x \rangle -\langle \zeta_2  , \mat z \rangle.
		\label{eq:0}
	\end{equation}
	It has been proved in our unpublished manuscript~\cite{chengwongold} that $f(\mat x)-f(\mat z) = -\frac{1}{2}\|\mat A (\mat z - \mat x)\|^2 - \langle  \mat \nabla f(\mat x), \mat z - \mat x\rangle$.  We give the proof below for completeness.  For all $s \in [0,1]$, define $\mat{z}_s = \mat{x} + s(\mat{z}-\mat{x})$.  By the chain rule, we have $\frac{\partial f}{\partial s} = \bigl\langle \frac{\partial f}{\partial \mat{z}_s},\frac{\partial \mat{z}_s}{\partial s} \bigr\rangle
	= \bigl\langle \nabla f(\mat{z}_s), \, \mat{z} -\mat{x} \bigr\rangle$.
	We integrate along a linear movement from $\mat{x}$ to $\mat{z}$.  Using the fact that $\nabla f(\mat{z}_s) = \mat{A}^t\mat{A}(\mat{x} + s(\mat{z}-\mat{x})) - \mat{A}^t \mat b = \nabla f(\mat{x}) + s\mat{A}^t\mat{A}(\mat{z}-\mat{x})$, we obtain $f(\mat{z}) = f(\mat{x}) + \int_{0}^1 \langle \nabla f(\mat{z}_s),  \mat{z} -\mat{x} \rangle \, \mathtt{d}s = f(\mat{x}) + \int_0^1 \langle \nabla f(\mat{x}), \, \mat{z}-\mat{x} \rangle \, \mathtt{d}s + 
	\int_0^1 s\iprod{\mat{A}^t\mat{A}(\mat{z}-\mat{x}), \, \mat{z}-\mat{x}} \, \mathtt{d}s \nonumber = f(\mat{x}) + \bigl[\langle \nabla f(\mat{x}), \, \mat{z}-\mat{x} \rangle \cdot s\bigr]^1_0 + 
	\bigl[\frac{1}{2}\norm{\mat{A}(\mat{z}-\mat{x})}^2 \cdot s^2 \bigr]^1_0 = f(\mat{x}) + \langle \nabla f(\mat{x}), \, \mat{z}-\mat{x} \rangle + \frac{1}{2}\norm{\mat{A}(\mat{z}-\mat{x})}^2$.	It follows immediately that $f(\mat{x}) - f(\mat{z}) =  -\langle \nabla f(\mat{x}), \, \mat{z}-\mat{x} \rangle - \frac{1}{2}\norm{\mat{A}(\mat{z}-\mat{x})}^2$.  By~\eqref{eq:0}, we can add $g(\mat{x}) - g(\mat{z})$ to the left side of this equation and $\langle \zeta_1  , \mat x \rangle -\langle \zeta_2  , \mat z \rangle$ to the right side.  We get $F(\mat x)-F(\mat z) = -\frac{1}{2}\|\mat A (\mat z - \mat x)\|^2 - \langle  \mat \nabla f(\mat x), \mat z - \mat x\rangle +\langle \zeta_1  , \mat x \rangle -\langle \zeta_2  , \mat z \rangle  = -\frac{1}{2}\|\mat A (\mat z - \mat x)\|^2- \langle  \mat \nabla f(\mat x), \mat z - \mat x\rangle +\langle \zeta_1  , \mat x \rangle -\langle \zeta_2 , \mat z \rangle  -\langle \zeta_1  , \mat z \rangle +\langle \zeta_1  , \mat z \rangle = -\frac{1}{2}\|\mat A (\mat z - \mat x)\|^2- \langle  \mat \nabla f(\mat x) + \zeta_1, \mat z - \mat x\rangle +\langle  \mat \zeta_1 - \zeta_2  , \mat z \rangle$.  This completes the proof of \eqref{eq:-1}. 
	
	Consider (iii).  By \eqref{eq:-1} with $\mat{x} = \mat{x}_r$ and $\zeta_1 = \zeta_{\mat{x}_r}$, we get 
	\begin{align}
		F(\mat x_r)-F(\mat z) 
		&=~-\frac{1}{2}\|\mat A (\mat  z - \mat x_r)\|^2- \langle  \mat \nabla f(\mat x_r) + \zeta_{\mat{x}_r}, \mat  z - \mat x_r\rangle +\langle  \mat \zeta_{\mat x_r}- \zeta_2, \mat  z \rangle \nonumber \\
		&\le~-\langle  \mat \nabla f(\mat x_r) + \zeta_{\mat{x}_r}, \mat  z - \mat x_r\rangle +\langle  \mat \zeta_{\mat x_r}- \zeta_2  , \mat  z \rangle. \label{eq:1}
	\end{align}
	Take any $i\in [n]$. By Lemma~\ref{g_subgradient-t}(i)(a), if $(\mat{z})_i > 0$, then $(\zeta_2)_i = \eta \geq (\zeta_{\mat{x}_r})_i$, and if $(\mat{z})_i < 0$, then $(\zeta_2)_i = -\eta \leq (\zeta_{\mat{x}_r})_i$.  As a result, $(\zeta_{\mat{x}_r} - \zeta_2)_i \cdot (\mat{z})_i \leq 0$ for all~$i$, proving that $\langle \zeta_{\mat{x}_r} -\zeta_2, \mat{z} \rangle \leq 0$.  Substituting $\langle \zeta_{\mat{x}_r} -\zeta_2, \mat{z} \rangle \leq 0$ into \eqref{eq:1} gives $F(\mat x_r)-F(\mat z) \le - \langle  \mat \nabla f(\mat x_r) + \zeta_{\mat{x}_r}, \mat  z - \mat x_r\rangle = - \langle  \gamma_{\mat x_r}, \mat  z - \mat x_r\rangle$.  This proves (iii).
	
	Consider (iv). Let $\xi$ be any subgradient in $\partial g(\mat{y}_r)$ that satisfies Lemma~\ref{gfy_to_gamma_y-t}(i).  By \eqref{eq:-1} with $\mat x=\mat y_r$, $\zeta_1 = \xi$, $\mat z=\mat x_r$, and $\zeta_2 = \zeta_{\mat{x}_r}$, we have $F(\mat{x}_r) - F(\mat{y}_r) = \frac{1}{2}\norm{\mat{A}(\mat{x}_r - \mat{y}_r)}^2 + \langle \nabla f(\mat{y}_r) + \xi, \mat{x}_r - \mat{y}_r \rangle  - \langle \xi - \zeta_{\mat{x}_r}, \mat{x}_r\rangle$.
	The middle term vanishes by Lemma~\ref{gfy_to_gamma_y-t}(i).  Therefore,
	\begin{equation}
		F(\mat x_r)-F(\mat y_r) = \frac{1}{2}\|\mat A (\mat x_r - \mat y_r)\|^2 +\langle \mat \zeta_{\mat x_r}-\xi  , \mat x_r \rangle.  \label{eq:2}
	\end{equation}
	By \eqref{eq:-1} again with
	$\mat x=\mat x_r$, $\zeta_1 = \zeta_{\mat{x}_r}$, $\mat z=\mat y_r$, and $\zeta_2 = \xi$.  It gives
	$F(\mat x_r)-F(\mat y_r)  = -\frac{1}{2}\|\mat A (\mat x_r - \mat y_r)\|^2- \langle  \mat \nabla f(\mat x_r) + \zeta_{\mat{x}_r}, \mat y_r - \mat x_r\rangle +\langle  \mat \zeta_{\mat x_r}- \mat \xi  , \mat y_r \rangle$.
	Summing the above equation and (\ref{eq:2}) gives:
	\begin{align*}
		2F(\mat x_r)-2F(\mat y_r) 
		&= -\langle  \mat \nabla f(\mat x_r) + \zeta_{\mat{x}_r}, \mat y_r - \mat x_r\rangle +\langle  \mat \zeta_{\mat x_r}- \mat \xi  , \mat y_r \rangle  +\langle \mat \zeta_{\mat x_r}-\xi  , \mat x_r \rangle \\
		&= -\langle  \mat \nabla f(\mat x_r) + \zeta_{\mat{x}_r}, \mat y_r - \mat x_r\rangle \tag*{\text{($\because$ Lemma~\ref{gfy_to_gamma_y-t}(ii))}} \\
		&= -\langle \gamma_{\mat{x}_r},\mat{y}_r-\mat{x}_r \rangle.
	\end{align*}
	This completes the proof of (iv).\hfill\qed
\end{proof}

\begin{lemma}
\label{lem:full-t}
Let $\mat{n}_r$ be any unit conical combination of $\{-\mathrm{sign}((\nabla f(\mat x_r))_i) \cdot \mat{e}_i : i \in E_r\}$.  Let $\mat{y}_r$ be the point in direction $\mat{n}_r$ from $\mat{x}_r$ that minimizes $F$. 
Let $\mat{n}_* = (\mat{x}_* - \mat{x}_r)/\norm{\mat{x}_* - \mat{x}_r}$.  If $F(\mat{x}_r) > F(\mat{x}_*) + \eps\norm{\mat{x}_*-\mat{x}_r}^2$~~for some $\eps \in (0,1)$, then $\langle \gamma_{\mat x_r} , \mat{y}_r - \mat{x}_r \rangle / \langle \gamma_{\mat x_r} , \mat{x}_* - \mat{x}_r \rangle \geq   \eps \cdot \langle \gamma_{\mat x_r} , \mat{n}_r \rangle^2 / \langle \gamma_{\mat x_r} , \mat{n}_* \rangle^2$.
\end{lemma}	 
\begin{proof}
By Lemma~\ref{gfy_to_gamma_y-t}(i), there exists $\xi \in \partial g(\mat{y}_r)$ such that $\langle  \mat \nabla f(\mat y_r) + \xi, \mat{n}_r\rangle = 0$.  We have $\langle \zeta_{\mat{x}_r},\mat{y}_r - \mat{x}_r \rangle - \langle \xi,\mat{y}_r - \mat{x}_r \rangle = \langle \zeta_{\mat{x}_r} - \xi, \mat{y}_r \rangle - \langle \zeta_{\mat{x}_r} - \xi, \mat{x}_r \rangle$ which is zero by Lemma~\ref{gfy_to_gamma_y-t}(ii).  It implies that $\langle \zeta_{\mat{x}_r},\mat{n}_r \rangle = \langle \xi,\mat{n}_r\rangle$, and hence $\langle  \mat \nabla f(\mat y_r) + \zeta_{\mat{x}_r}, \mat{n}_r\rangle = 0$.  Substituting $\nabla f(\mat y_r)$ by $\mat A^t \mat A \mat y_r -  \mat A^t \mat b$, we obtain $\langle \mat{A}^t\mat{A} \mat{y}_r -\mat{A}^t\mat{b} + \zeta_{\mat{x}_r}, \mat{n}_r\rangle  = 0$. Rearranging terms gives $\langle \mat{A}^t\mat{A}(\mat{y}_r - \mat{x}_r), \mat{n}_r \rangle = \langle - \mat{A}^t\mat{A}\mat{x}_r + \mat{A}^t\mat{b} - \zeta_{\mat{x}_r},\mat{n}_r\rangle 
= \langle  -\nabla f(\mat x_r)- \zeta_{\mat x_r}, \mat{n}_r \rangle = \langle -\gamma_{\mat x_r} ,\mat{n}_r \rangle$.
Therefore, $\norm{\mat{y}_r-\mat{x}_r} \cdot \norm{\mat{A}\mat{n}_r}^2 = \langle \mat{A}^t\mat{A}(\mat{y}_r - \mat{x}_r), \mat{n}_r \rangle = \langle -\gamma_{\mat x_r} ,\mat{n}_r \rangle$.  Hence, $\norm{\mat{y}_r - \mat{x}_r} \geq \langle -\gamma_{\mat x_r} ,\mat{n}_r \rangle$ as $\norm{\mat{A}\mat{n}_r}^2 \leq \norm{\mat{A}}^2 \leq 1$.

By Lemma~\ref{gfy_to_gamma_y-t}(iii), $F(\mat{x}_r) - F(\mat{x}_*) \leq \norm{\mat{x}_*-\mat{x}_r} \cdot \langle -\gamma_{\mat{x}_r}, \mat{n}_*\rangle$.  If $\langle -\gamma_{\mat{x}_r}, \mat{n}_*\rangle < \eps\norm{\mat{x}_*-\mat{x}_r}$, then $F(\mat{x}_r) - F(\mat{x}_*) \leq \eps\norm{\mat{x}_*-\mat{x}_r}^2$, contradicting the assumption of the lemma.  Hence, $\langle -\gamma_{\mat{x}_r}, \mat{n}_*\rangle \geq \eps\norm{\mat{x}_*-\mat{x}_r}$.  Combining this inequality with $\norm{\mat{y}_r - \mat{x}_r} \geq \langle -\gamma_{\mat x_r} ,\mat{n}_r \rangle$ gives
$\frac{\langle \gamma_{\mat{x}_r}, \mat{y}_r - \mat{x}_r \rangle}{\langle \gamma_{\mat{x}_r}, \mat{x}_* - \mat{x}_r \rangle} = \frac{\norm{\mat{y}_r-\mat{x}_r}}{\norm{\mat{x}_*-\mat{x}_r}} \cdot \frac{\langle \gamma_{\mat{x}_r},\mat{n}_r\rangle}{\langle \gamma_{\mat{x}_r},\mat{n}_*\rangle} \geq \eps \cdot \frac{\langle \gamma_{\mat{x}_r},\mat{n}_r\rangle^2}{\langle \gamma_{\mat{x}_r},\mat{n}_*\rangle^2}$.\hfill\qed
\end{proof}

\begin{lemma}
\label{lem:zero}
Let $G_r$ be the set of the $\tau_{r+1}$ heaviest elements in $E_r$.  Let $H_r = G_r \cup (\supp(\mat{x}_*) \cap E_r)$.
Then, $\langle \gamma_{\mat{x}_r}, \mat{x}_* - \mat{x}_r \rangle = \langle \gamma_{\mat{x}_r} \! \downarrow \! H_r, \mat{x}_* - \mat{x}_r \rangle$.
\end{lemma}
\begin{proof}
	For any $i \in H_r$, $(\gamma_{\mat{x}_r})_i = (\gamma_{\mat{x}_r} \! \downarrow \! H_r)_i$ by definition.  Therefore, $(\gamma_{\mat{x}_r})_i \cdot (\mat{x}_* - \mat{x}_r)_i = (\gamma_{\mat{x}_r} \! \downarrow \! H_r)_i \cdot (\mat{x}_* - \mat{x}_r)_i$.
	
	Take any $i \not\in E_r$.  We have $(\gamma_{\mat{x}_r})_i = 0$ because $\gamma_{\mat{x}_r} = \zeta_{\mat{x}_r} + \nabla f(\mat{x}_r)$ and $(\zeta_{\mat{x}_r})_i = -(\nabla f(\mat{x}_r))_i$  by definition. Therefore, both $(\gamma_{\mat{x}_r})_i \cdot (\mat{x}_* - \mat{x}_r)_i$ and $(\gamma_{\mat{x}_r} \! \downarrow \! H_r)_i \cdot (\mat{x}_* - \mat{x}_r)_i$ are zero.
	
	For any $i \in E_r \setminus H_r$, $i \not\in \supp(\mat{x}_*)$ as $\supp(\mat{x}_*) \cap E_r \subseteq H_r$.  So $(\mat{x}_*)_i = 0$.  Also, $(\mat{x}_r)_i = 0$ as $i \in E_r$.  So $(\gamma_{\mat{x}_r})_i \cdot (\mat{x}_* - \mat{x}_r)_i$ and $(\gamma_{\mat{x}_r} \! \downarrow \! H_r)_i \cdot (\mat{x}_* - \mat{x}_r)_i$ are zero.\hfill\qed
\end{proof}

\begin{lemma}
\label{lem:full_cr-t}
If $F(\mat{x}_r) > F(\mat{x}_*) + \eps \norm{\mat{x}_*-\mat{x}_r}^2$, then 
\[
\frac{F(\mat{x}_{r+1}) - F(\mat{x}_*)}{F(\mat{x}_r) - F(\mat{x}_*)} \leq 1 - \frac{\eps\tau_{r+1}}{8(s+\tau_{r+1})\ln\tau_{r+1}}.
\]
\end{lemma}	
\begin{proof}
Let $G_r$ be the set of the $\tau_{r+1}$ heaviest elements in $E_r$.  Let $H_r = G_r \cup (\supp(\mat{x}_*) \cap E_r)$.
We prove in Lemma~\ref{lem:tech-nice} in Appendix~\ref{app:nice} 
that there exists a unit descent direction ${\mat{n}}_r$ from $\mat{x}_r$ such that $\mat{n}_r$ is a conical combination of $\{-\mathrm{sign}((\nabla f(\mat x_r))_i) \cdot \mat{e}_i : i \in G_r\}$ and $\langle -\gamma_{\mat x_r} ,{\mat{n}}_r \rangle \geq \norm{\gamma_{\mat x_r} \! \downarrow \! H_r} \cdot \sqrt{\frac{\tau_{r+1}}{4(s+\tau_{r+1})\ln\tau_{r+1}}}$.  Let $\mat{n}_* = (\mat{x}_* - \mat{x}_r)/\norm{\mat{x}_* - \mat{x}_r}$.   By Lemma~\ref{lem:zero}, $\langle -\gamma_{\mat x_r} ,\mat{n}_* \rangle =
\langle -\gamma_{\mat x_r} \! \downarrow \! H_r,\mat{n}_* \rangle$.  Then, $\langle -\gamma_{\mat x_r},\mat{n}_* \rangle \leq \norm{\gamma_{\mat x_r} \! \downarrow \! H_r}$.  Hence, 
$\frac{\langle \gamma_{\mat{x}_r},\mat{n}_r\rangle}{\langle \gamma_{\mat{x}_r},\mat{n}_*\rangle} \geq \sqrt{\frac{\tau_{r+1}}{4(s+\tau_{r+1})\ln\tau_{r+1}}}$.

Let $\mat{y}_r$ be the point in the direction $\mat{n}_r$ from $\mat{x}_r$ that minimizes $F$. We have
$\frac{F(\mat{x}_{r+1}) - F(\mat{x}_*)}{F(\mat{x}_r) - F(\mat{x}_*)} = 1 - \frac{F(\mat{x}_r) - F(\mat{x}_{r+1})}{F(\mat{x}_r) - F(\mat{x}_*)} 
\leq 1 - \frac{F(\mat{x}_r) - F(\mat{y}_r)}{F(\mat{x}_r) - F(\mat{x}_*)}$.   By Lemma~\ref{gfy_to_gamma_y-t}(iii) and (iv), $F(\mat{x}_r) - F(\mat{y}_r) = -\frac{1}{2}\langle \gamma_{\mat{x}_r}, \mat{y}_r - \mat{x}_r \rangle$ and $F(\mat{x}_r) - F(\mat{x}_*) \leq -\langle \gamma_{\mat{x}_r}, \mat{x}_* - \mat{x}_r \rangle$. Therefore, $\frac{F(\mat{x}_{r+1}) - F(\mat{x}_*)}{F(\mat{x}_r) - F(\mat{x}_*)} \leq 1 - \frac{1}{2} \cdot \frac{\langle \gamma_{\mat{x}_r},\mat{y}_r-\mat{x}_r\rangle}{\langle \gamma_{\mat{x}_r},\mat{x}_*-\mat{x}_r\rangle} \leq 1- \frac{\eps}{2} \cdot \frac{\langle \gamma_{\mat{x}_r},\mat{n}_r\rangle^2}{\langle \gamma_{\mat{x}_r},\mat{n}_*\rangle^2}$ by Lemma~\ref{lem:full-t}, which is at most $1 - \frac{\eps\tau_{r+1}}{8(s+\tau_{r+1})\ln\tau_{r+1}}$ by the lower bound on $\frac{\langle \gamma_{\mat{x}_r},\mat{n}_r\rangle}{\langle \gamma_{\mat{x}_r},\mat{n}_*\rangle}$.\hfill\qed
\end{proof}

\cancel{
\begin{lemma}
\label{lem:gap-ratio-t}
If $\mat{x}_* \not= 0$, then $F(\mat{x}_*) \geq \eta\norm{\mat{b}}/4$.
\end{lemma}
\begin{proof}
	If $\norm{\mat{Ax}_*} \leq \norm{\mat{b}}/4$, then $F(\mat{x}_*) \geq \frac{1}{2}\norm{\mat{Ax} - \mat{b}}^2 \geq \frac{9}{32}\norm{\mat{b}}^2$.  If $\norm{\mat{Ax}_*} > \norm{\mat{b}}/4$, then $F(\mat{x}_*) \geq \eta\norm{\mat{x}_*}_1 \geq \eta\norm{\mat{x}_*} \geq \eta\norm{\mat{Ax}_*} > \eta\norm{\mat{b}}/4$.   Recall that $\eta < \norm{\mat{A}^t\mat{b}}_\infty \leq \norm{\mat{b}}$.  Therefore, $\eta\norm{\mat{b}}/4 < \frac{9}{32}\norm{\mat{b}}^2$.\hfill\qed
\end{proof}
}

Recall the parameter $h \in (1,2]$ in Algorithm~\ref{alg:1}.  The next result bounds the working set sizes up to the first solution with an additive error at most $\eps/\eta^2$.

\begin{theorem}
\label{thm:tau-sum}
Suppose that $\norm{\mat{A}} \leq 1$ and $\eta = \alpha\norm{\mat{A}^t\mat{b}}_\infty$ for a fixed $\alpha \in (0,1)$.  Scale space such that $\norm{b} = 1$.
Let $\kappa+1$ be the minimum index such that $F(\mat{x}_{\kappa+1}) - F(\mat{x}_*) \leq \eps/\eta^2$.
\begin{itemize}
\item If $h \leq 2^{O(\eps/(\ln n \ln (\eta/\eps)))}$, then $\sum_{i=1}^\kappa \tau_i = O\left(\frac{1}{\eps}(s+\tau)\log (s+\tau)\log\frac{\eta}{\eps}\right)$.
\item Otherwise, $\sum_{i=1}^{\kappa} \tau_i = O\left(\frac{1}{\eps}k\log k\log\frac{\eta}{\eps}\right)$.
\end{itemize}
\end{theorem}
\begin{proof}
Take any constant $c \geq 1$.  Divide $[\kappa]$ into two disjoint subsets $I$ and $J$ such that for all $i \in I$, $\tau_i \leq cs$, and for all $i \in J$, $\tau_i > cs$.

For any $\mat{x}_r$, $\eta\norm{\mat{x}_r} \leq \eta\norm{\mat{x}_r}_1 \leq F(\mat{x}_r) \leq F(\mat{x}_0) = \frac{1}{2}\norm{\mat{b}}^2 = \frac{1}{2}$.   The same argument works for $\mat{x}_*$.  It means that $\eps\norm{\mat{x}_* - \mat{x}_r}^2 \leq \eps/\eta^2$.  Therefore, for any $r \in [\kappa]$, $F(\mat{x}_r) - F(\mat{x}_*) > \eps/\eta^2 \geq \eps\norm{\mat{x}_*-\mat{x}_r}^2$ by assumption, which makes Lemma~\ref{lem:full_cr-t} applicable for all $r \in [\kappa]$.

View $I$ as a chronological sequence.  Let $i$ be the largest index in $I$.   Note that $F(\mat{x}_0) - F(\mat{x}_*) \leq F(\mat{x}_0) \leq \frac{1}{2}\norm{\mat{b}}^2 = \frac{1}{2}$.  Then, by Lemma~\ref{lem:full_cr-t}, $F(\mat{x}_{i}) - F(\mat{x}_*) \leq \frac{1}{2}\prod_{i \in I} \bigl(1-\frac{\eps\tau_i}{8(s+\tau_i)\ln (cs)}\bigr)$, which is at most $\frac{1}{2}\prod_{i \in I} \bigl(1-\frac{\eps\tau_i}{8(c+1)s\ln(cs)}\bigr)$.  Let $\tau_{\text{avg}} = \sum_{i\in I} \tau_i/|I|$.  It is well known that the geometric mean is at most the arithmetic mean.  Therefore, $\frac{1}{2}\prod_{i \in I} \bigl(1-\frac{\eps\tau_i}{8(c+1)s\ln(cs)}\bigr) \leq 
\frac{1}{2}\bigl(1 - \frac{\eps\tau_{\text{avg}}}{8(c+1)s\ln(cs)}\bigr)^{|I|} \leq \frac{1}{2}e^{-\eps\tau_{\text{avg}}|I|/(8(c+1)s\ln(cs))}$. This upper bound is at least $\eps/\eta^2$ so that $\mat{x}_{\kappa+1}$ is the first solution that satisfies $F(\mat{x}_{\kappa+1}) - F(\mat{x}_*) \leq \eps/\eta^2$.  Hence, $\frac{\eps\tau_{\text{avg}}|I|}{8(c+1)s\ln(cs)} \leq \ln \frac{\eta^2}{2\eps}$, which implies that $\sum_{i\in I} \tau_i = \tau_{\text{avg}}|I| = O\left(\frac{1}{\eps}s\log s \log\frac{\eta}{\eps}\right)$.

View $J$ as a chronological sequence.  Let $\tau_{\max} = \max_{i \in J} \tau_i$.  Take a contiguous subsequence of $J$ of length $(16\ln \tau_{\max})/\eps$.  Let $i$ and $j$ be the minimum and maximum indices in this subsequence, respectively.  By Lemma~\ref{lem:full_cr-t}, $F(\mat{x}_j) - F(\mat{x}_*) \leq e^{-1} \cdot (F(\mat{x}_i) - F(\mat{x}_*))$.  Since $F(\mat{x}_0) - F(\mat{x}_*) \leq F(\mat{x}_0) = \frac{1}{2}\norm{\mat{b}}^2 = \frac{1}{2}$, we can divide $J$ into no more than $\ln \frac{\eta^2}{2\eps}$ contiguous subsequences of length $(16\ln \tau_{\max})/\eps$.  It follows that $|J| \leq \frac{16}{\eps}\ln \tau_{\max} \cdot \ln \frac{\eta^2}{2\eps}$.  
The algorithm ensures that $\tau_{i+1} \leq h\tau_i$.  Extract the longest subsequence of $J$ (not necessarily contiguous) in which $\tau_i$ strictly increases.  Every consecutive $\tau_i$'s in this subsequence differ by a factor $h$.  This subsequence starts with $\min_{i \in J} \tau_i \leq \max\{hcs,\tau\}$.  If $h\leq 2^{O(\eps/(\ln n \ln (\eta/\eps)))}$, then $\tau_{\max} \leq h^{|J|} \cdot \max\{hcs,\tau\} \leq 2^{O(1)} \cdot (s+\tau)$.  So $\sum_{i \in J} \tau_i \leq O(s+\tau) \cdot |J| = O\left(\frac{1}{\eps}(s+\tau)\log (s+\tau)\log\frac{\eta}{\eps}\right)$.
If $h > 2^{O(\eps/(\ln n \ln (\eta/\eps)))}$, we still have $\tau_{\max} \leq k$ and hence $\sum_{i \in J} \tau_i \leq k|J| = O\left(\frac{1}{\eps}k\log k \log\frac{\eta}{\eps}\right)$.\hfill\qed
%
\end{proof}



\noindent {\em Remark 1.}  Clearly $\kappa \leq \sum_{i=1}^\kappa \tau_i$.  One can work out the exact upper bound for $\sum_{i=1}^\kappa \tau$ and hence $\kappa$.  DWS can be stopped after $\kappa$ iterations to obtain an error at most $\eps/\eta^2$, although DWS has probably terminated earlier in practice. 

\vspace{6pt}

\noindent {\em Remark 2.} Starting from $p_0$, we add at most $\sum_{i=1}^\kappa \tau_i$ free variables to any working set before reaching $\mat{x}_{\kappa+1}$.  Suppose that we set $p_0$ and $\tau$ to be $O(1)$.  If $\eps$ is given beforehand, we can ensure that every working set has $O(\frac{1}{\eps} s\log s \log \frac{\eta}{\eps})$ variables before reaching $\mat{x}_{\kappa}$.  So each call of the solver runs provably faster than using all $n$ variables.  When $\eps$ is not given, if $k = \Theta(s \log (n/s))$ (sufficient for the true signal to be recovered with high probability), every working set still has only $O(\frac{1}{\eps} s \cdot \mathrm{polylog}(n))$ variables.
%
%
Clearly, $|\supp(\mat{x}_r)| \leq |W_r|$.  It follows that $|\supp(\mat{x}_{r})| \leq p_0 + \sum_{i=1}^\kappa \tau_i$.  
We conclude that all solutions $\mat{x}_r$, $r \in [\kappa+1]$, are provably sparse if $\eps$ is given beforehand or $k = \Theta(s \log (n/s))$.

\vspace{6pt}

\noindent {\em Remark 3.}  Figure~\ref{fig:sota_supp_wss_12} shows that the working set sizes are at most $cs$ for some small constant $c$ in the experiments.  That is, $\max_{i \in [\kappa]} \tau_i = O(s)$.  Under this assumption, the proof of Theorem~\ref{thm:tau-sum} reveals that $\sum_{i=1}^\kappa \tau_i = O(\frac{1}{\eps} s \log s \log \frac{\eta}{\eps})$ even if $\eps$ is not given beforehand.  Then, the working set sizes and support set sizes can be bounded by $O(\frac{1}{\eps} s \log s \log \frac{\eta}{\eps})$ even if $\eps$ is not given beforehand. 

\vspace{6pt}

\noindent {\em Remark 4.} 
To prepare for the next iteration, we need to compute $\nabla f(\mat{x}_r) = \mat{A}^t\mat{A}\mat{x}_r - \mat{A}^t\mat{b}$. We precompute $\mat{A}^t\mat{b}$ in $O(kn)$ time. Let $w = |W_r| \leq \sum_{i=1}^\kappa \tau_i$.  Note that $|\supp(\mat{x}_r)| \leq w$.  To obtain $\mat{A}^t\mat{A}\mat{x}_r$, we use $\mat{A}_r \in \real^{k \times w}$ and $\supp(\mat{x}_r)$ to obtain $\mat{A}\mat{x}_r$ in $O(kw)$ time, and then we compute $\mat{A}^t (\mat{A}\mat{x}_r)$ in $O(kn)$ time.  We extract $\mat{A}_{r+1}$ from $\mat{A}$ corresponding to $W_{r+1}$ is $O(k|W_{r+1}|) = O(k \cdot \sum_{i=1}^\kappa \tau_i)$ time.

\vspace{6pt}

\noindent {\em Remark 5.}
		If $\norm{\mat{Ax}_*} \leq \norm{\mat{b}}/4$, then $F(\mat{x}_*) \geq \frac{1}{2}\norm{\mat{Ax} - \mat{b}}^2 \geq \frac{9}{32}\norm{\mat{b}}^2$.  If $\norm{\mat{Ax}_*} > \norm{\mat{b}}/4$, then $F(\mat{x}_*) \geq \eta\norm{\mat{x}_*}_1 \geq \eta\norm{\mat{x}_*} \geq \eta\norm{\mat{Ax}_*} > \eta\norm{\mat{b}}/4$.   Recall that $\eta < \norm{\mat{A}^t\mat{b}}_\infty \leq \norm{\mat{b}}$.   We conclude that $F(\mat{x}_*) \geq \eta\norm{b}/4$ which is at least $\eta/4$ after $\norm{\mat{b}}$ is scaled to 1.  Therefore, the additive error of $\eps/\eta^2$ in Theorem~\ref{thm:tau-sum} is at most $\frac{4\eps}{\eta^3} F(\mat{x}_*)$.

\bibliographystyle{splncs04}
\bibliography{mybibliography}


\appendix




\newpage

\section{URLs to the Solver Pages}\label{urls}
Here is a list of URLs to the solver pages:
\begin{itemize}
    \item GPSR: \href{http://www.lx.it.pt/~mtf/GPSR/}{http://www.lx.it.pt/$\sim$mtf/GPSRR/}
    \item Celer: \href{https://github.com/mathurinm/celer}{https://github.com/mathurinm/celer}
    \item Skglm: \href{https://github.com/scikit-learn-contrib/skglm}{https://github.com/scikit-learn-contrib/skglm}
    \item Benchopt: \href{https://github.com/benchopt/benchopt}{https://github.com/benchopt/benchopt}
\end{itemize}
Note that these links are provided only for courtesy purposes. The authors do not have any direct or indirect control over the public pages and are unaffiliated.

\section{Existence of a good descent direction}
\label{app:nice}

Define the following subset of $E_r$:
\begin{align*}
H_r = \bigl\{ i \in E_r :&~~i \in \supp(\mat{x}_*) \,\,\, \vee \, \\ 
&~~\text{$i$ is one of the $\tau_{r+1}$ heaviest elements in $E_r$} \bigr\}.
\end{align*}
Recall that for any vector $\gamma \in \real^n$, $\gamma \! \downarrow 
\! H_r$ denotes the projection of $\gamma$ in the linear subspace spanned by $\{\mat{e}_i : i \in H_r\}$.

For all $i \in E_r$, let $\mat{s}_i = -\mathrm{sign}((\nabla f(\mat{x}_r))_i) \cdot \mat{e}_i$.  Recall from Lemma~\ref{lem:32} that every conical combination of $\{\mat{s}_i : i \in E_r\}$ is a descent direction from $\mat{x}_r$.

\begin{lemma}
\label{lem:tech-1}
  For any $\alpha \in (0,1]$, there exists $j$ among the $t=\alpha |H_r|$ heaviest elements in $H_r$ such that for every $i \in H_r$, if the weight of $i$ is at least the weight of $j$, then $\langle \gamma_{\mat x_r} \! \downarrow \! H_r,\mat{s}_i \rangle^2 \geq \norm{\gamma_{\mat x_r}  \! \downarrow \! H_r}^2 \cdot \alpha/(2j\ln t)$.
\end{lemma}
\begin{proof}
Consider a histogram $T_1$ of $\alpha/(2i\ln t)$ against $i \in [t]$.   The total length of the vertical bars in $T_1$ is $\sum_{i=1}^t \alpha/(2i\ln t) \leq \alpha$ as $\sum_{i=1}^t 1/i \leq 1 + \ln t \leq 2\ln t$.

Consider another histogram $T_2$ of $\langle \gamma_{\mat{x}_r}    \! \downarrow \! H_r,\mat{s}_i\rangle^2/\norm{\gamma_{\mat{x}_r}  \! \downarrow \! H_r}^2$ against $i \in H_r$.  By Lemma~\ref{g_subgradient-t}(ii), $-\mathrm{sign}((\gamma_{\mat{x}_r})_i) = \mathrm{sign}(\mat{s}_i)$ for all $i \in H_r$.  Therefore, $-\gamma_{\mat{x}_r}  \! \downarrow \! H_r$ is a conical combination of $\{\mat{s}_i : i \in H_r \}$.  It follows that the total length of the vertical bars in $T_2$, which is $\sum_{i\in H_r} \langle \gamma_{\mat x_r} \! \downarrow \! H_r,\mat{s}_i\rangle^2/\norm{\gamma_{\mat{x}_r} \! \downarrow \! H_r}^2$, is equal to 1.

For $i \in E_r$, $(\gamma_{\mat{x}_r})_i = (\nabla f(\mat{x}_r))_i + (\zeta_{\mat{x}_r})_i$, $(\zeta_{\mat{x}_r})_i = -\mathrm{sign}((\nabla f(\mat{x}))_i) \cdot \eta$, and $|(\nabla f(\mat{x}_r))_i| > \eta$.  Therefore, $(\gamma_{\mat{x}_r})_i = \mathrm{sign}((\nabla f(\mat{x}_r))_i) \cdot (|(\nabla f(\mat{x}_r))_i| - \eta)$, which implies that $|(\gamma_{\mat{x}_r})_i| = |(\nabla f(\mat{x}_r))_i| - \eta$.  Hence, the $\tau_{r+1}$ heaviest elements of $E_r$ are also the elements of $E_r$ with the $\tau_{r+1}$ largest $|(\gamma_{\mat{x}_r})_i|$'s. 
Consequently, the total length of the vertical bars in $T_2$ for the first $t=\alpha |H_r|$ indices is at least $\alpha$.  

There must be an index $j$ among the $t$ heaviest elements of $H_r$ such that the vertical bar in $T_2$ at $j$ is not shorter than the vertical bar in $T_1$ at $j$.  That is, for every $i \in H_r$, if the weight of $i$ is at least the weight of $j$, then $\langle \gamma_{\mat x_r} \! \downarrow \! H_r,\mat{s}_i \rangle^2/\norm{\gamma_{\mat x_r}  \! \downarrow \! H_r}^2 \geq \alpha/(2j\ln t)$.\hfill\qed
\end{proof}

Let $G_r$ be the subset of the $\tau_{r+1}$ heaviest elements in $E_r$, which are also the $\tau_{r+1}$ heaviest elements in $H_r$.  Using $\alpha = \frac{\tau_{r+1}}{s+\tau_{r+1}}$, Lemma~\ref{lem:tech-1} implies that for every $i \in G_r$, $\gamma_{\mat{x}_r} \! \downarrow \! \{i\}$ makes an angle no larger than $\arccos\left(1/\sqrt{2(s+\tau_{r+1})\ln\tau_{r+1}}\right)$ with $\gamma_{\mat{x}_r} \! \downarrow \! H_r$.  This is the basis that a descent direction can be obtained using a smaller subset $G_r$ of $E_r$.  This angle bound can be reduced using Lemma~\ref{lem:nice} below which is proved in our unpublished manuscript~\cite{chengwongold}.

\begin{figure}
	\centerline{\includegraphics[scale=0.65]{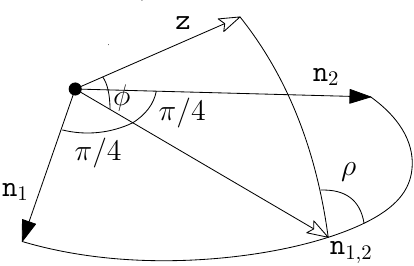}}
	\caption{The vector $\mat{n}_{1,2}$ bisects the right angle $\angle (\mat{n}_1,\mat{n}_2)$.  The angle $\rho$ is at least $\pi/2$.}
	\label{fg:nice}
\end{figure}

\begin{lemma}
	\label{lem:nice}
	Take any $c \leq 1/\sqrt{2}$.  Let $\mat{z}$ be a vector in $\real^D$ for some $D \geq 2$.  Suppose that there is a set $V$ of unit vectors in $\real^D$ such that the vectors in $V$ are mutually orthogonal, and for every $\mat{n} \in V$, $\cos \angle (\mat{n},\mat{z}) \geq c|V|^{-1/2}$.  There exists a conical combination $\mat{y}$ of the vectors in $V$ such that $\cos\angle (\mat{y},\mat{z}) \geq c/\sqrt{2}$.
\end{lemma}
\begin{proof}
	Let $\theta = \arccos\left(c |V|^{-1/2}\right)$.   If $\theta \leq \pi/3$, we can pick any vector $\mat{n} \in V$ as $\mat{y}$ because $\cos\angle(\mat{n},\mat{z}) \geq \cos\theta \geq \cos(\pi/3) \geq c/\sqrt{2}$ for any $c \leq 1/\sqrt{2}$.  Suppose that $\theta > \pi/3$.  Let $W$ be a maximal subset of $V$ whose size is a power of 2.  Arbitrarily label the vectors in $W$ as $\mat{n}_1, \mat{n}_2, \ldots$.   Consider the unit vector $\mat{n}_{1,2} = \frac{1}{\sqrt{2}}\mat{n}_1 + \frac{1}{\sqrt{2}}\mat{n}_2$.   Let $\phi = \angle (\mat{n}_{1,2},\mat{z})$.  Refer to Figure~\ref{fg:nice}.   By assumption, $\mat{n}_1 \perp \mat{n}_2$.  Let $\rho$ be the non-acute angle between the plane spanned by $\{\mat{n}_1,\mat{n}_2\}$ and the plane spanned by $\{\mat{n}_{1,2},\mat{z}\}$.  By the spherical law of cosines, $\cos\theta \leq \cos\angle (\mat{n}_2,\mat{z}) = \cos\phi\cos(\pi/4) + \sin\phi\sin(\pi/4)\cos\rho$.
	Note that $\cos\rho \leq 0$ as $\rho \geq \pi/2$.  So $\cos\phi \geq \sec(\pi/4) \cos\theta = \sqrt{2}\cos\theta$.
	The same analysis holds between $\mat{z}$ and the unit vector $\mat{n}_{3,4} = \frac{1}{\sqrt{2}}\mat{n}_3 + \frac{1}{\sqrt{2}}\mat{n}_4$, and so on.  So we obtain $|W|/2$ vectors $\mat{n}_{2i-1,2i}$ for $i = 1,\ldots,|W|/2$ such that $\angle (\mat{n}_{2i-1,2i},\mat{z}) \leq \arccos\bigl(\sqrt{2}\cos\theta\bigr)$.  Call this the first stage.   Repeat the above with the $|W|/2$ unit vectors $\mat{n}_{1,2}, \mat{n}_{3,4}, \ldots$ in the second stage and so on.  We end up with one vector in $\log_2 |W|$ stages.  If we produce a vector that makes an angle at most $\pi/3$ with $\mat{z}$ before going through all $\log_2 |W|$ stages, the lemma is true.  Otherwise, we produce a vector $\mat{y}$ in the end such that $\cos \angle (\mat{y},\mat{z}) \geq \bigl(\sqrt{2}\bigr)^{\log_2 |W|}\cos\theta \geq \bigl(\sqrt{2}\bigr)^{\log_2 |V|-1}\cos\theta  \geq \sqrt{|V|/2} \cdot \cos\theta = c/\sqrt{2}$.
\end{proof}

\begin{lemma}
\label{lem:tech-nice}
Let $G_r$ be the subset of the $\tau_{r+1}$ heaviest elements of $E_r$.
There exists a descent direction $\mat{n}_r$ from $\mat{x}_r$ such that $\mat{n}_r$ is a conical combination of $\{\mat{s}_i : i \in G_r\}$ and $\langle -\gamma_{\mat x_r} ,{\mat{n}}_r \rangle \geq \norm{\gamma_{\mat x_r} \! \downarrow \! H_r} \cdot \sqrt{\frac{\tau_{r+1}}{4(s+\tau_{r+1})\ln\tau_{r+1}}}$. 
\end{lemma}
\begin{proof}
Using $\alpha = \frac{\tau_{r+1}}{s+\tau_{r+1}}$, Lemma~\ref{lem:tech-1} implies that for every $i \in G_r$, $\gamma_{\mat{x}_r} \! \downarrow \! \{i\}$ makes an angle no larger than $\arccos\left(1/\sqrt{2(s+\tau_{r+1})\ln\tau_{r+1}}\right)$ with $\gamma_{\mat{x}_r} \! \downarrow \! H_r$. We apply Lemma~\ref{lem:nice} with $V = G_r$ and $c = \sqrt{\frac{\tau_{r+1}}{2(s+\tau_{r+1})\ln\tau_{r+1}}}$.  By Lemma~\ref{lem:nice}, there is a conical combination of $\{\mat{s}_i : i \in G_r\}$ that improves this angle bound to $\arccos\bigl(\sqrt{\frac{\tau_{r+1}}{4(s+\tau_{r+1})\ln\tau_{r+1}}}\bigr)$.
\end{proof}

\end{document}